\documentclass[AMS,STIX2COL]{WileyNJD-v2}

\articletype{Article Type}%

\received{26 April 2016}
\revised{6 June 2016}
\accepted{6 June 2016}

\usepackage{cases}
\usepackage{amsmath}

\allowdisplaybreaks[4]

\begin{document}

\title{Distributed Model Predictive Control for Asynchronous Multi-agent Systems with Self-Triggered Coordinator}

\author[1,2]{Qianqian Chen}

\author[1,2]{Shaoyuan Li}

\authormark{CHEN \textsc{et al}}

\address[1]{\orgdiv{Department of Automation}, \orgname{Shanghai Jiao Tong University}, \orgaddress{\state{Shanghai}, \country{China}}}

\address[2]{\orgdiv{Key Laboratory of System Control and Information processing}, \orgname{Shanghai Jiao Tong University}, \orgaddress{\state{Shanghai}, \country{China}}}

\corres{Shaoyuan Li, Department of Automation, Shanghai Jiao Tong University, Shanghai 200240, China.
\\ \email{syli@sjtu.edu.cn}}

\abstract[Abstract]{
This paper investigates the distributed model predictive control for an asynchronous nonlinear multi-agent system with external interference via a self-triggered generator and a prediction horizon regulator.
First, a shrinking constraint related to the error between the actual state and the predicted state is introduced into the optimal control problem to enable the robustness of the system.
Then, the trigger interval and the corresponding prediction horizon are determined by altering the expression of the Lyapunov function, thus achieving a trade-off between control performance and energy loss.
By implementing the proposed algorithm, the coordination objective of the multi-agent system is achieved under asynchronous communication.
Finally, the recursive feasibility and stability are proven successively.
An illustrative example is conducted to demonstrate the merits of the presented approach.
}



\maketitle

\section{Introduction}
Following the expansion of industrial systems, distributed model predictive control (DMPC) \cite{SVRWP-2010} emerges as a powerful framework for large-scale multi-agent systems (MASs) \cite{W-2009, ZLJL-2022}.
The essence of DMPC is allocating computing tasks to multiple agents for parallel processing.
Each agent adopts model predictive control technology and interacts with the neighbor agents, thus achieving the predetermined distributed targets.
DMPC possesses the superiority of improved scalability, favorable reliability, and high resource utilization \cite{GLZ-2020}. Therefore, this research field arouses widespread concern among scholars and produces many achievements for some practical systems.
For instance, the authors in \cite{ZGLXZL-2020} propose the optimal economic dispatch technique for islands' multi-microgrids based on dynamic non-cooperative game theory.
The minimum performance bound is devised in \cite{HL-2021} by DMPC algorithm for power flow systems. The fault-tolerant issue as a consequence of DMPC methodology is addressed in \cite{ZSLD-2022} for multi-unmanned aerial vehicle systems.
Additionally, existing theoretical research recognizes the critical role played by DMPC algorithm.
For example, the authors in \cite{DZQSX-2022} and \cite{DQSZX-2022} take the economic performance of the optimized system into account, respectively under the Lyapunov-based approach and terminal penalty approach.
The authors in \cite{CMFMC-2021} incorporate the idea of coalition, and achieve the time-varying degree of system collaboration by dynamically forming different clusters.
A novel robust adaptive terminal set is proposed in \cite{WM-2022} to expand the attraction region. The authors in \cite{HL-2023} utilize the topological information to construct an available architecture, which is conducive to the consequential DMPC coordination.
The authors in \cite{DGPW-2016} apply DMPC mechanism to achieve the tracking and formation objective under the time-varying topology.
Nevertheless, the above control strategies are \textit{time-based}, i.e., the optimal control problem (OCP) is recalculated at each instant.
The repetitive access leads to the waste of computing costs and communication resources.

To this end, \textit{event-based} control \cite{T-2007} technique is presented to achieve resource conservation. Under the framework of the event-based DMPC mechanism \cite{DFJ-2011}, OCP for each agent is solved merely when the conditions about control performance are violated, thus avoiding unnecessary calculation. Event-based DMPC can be classified into event-triggered DMPC \cite{ZSLNL-2019, ZLXG-2022} and self-triggered DMPC. In contrast to event-triggered DMPC, self-triggered DMPC avoids the consistent monitoring of the system and directly calculates the next instant to solve the OCP on account of the information known at the current instant. Extensive research shows that this control strategy can not only effectively reduce the computational burden, but also ensure feasibility and stability.
{\textcolor{blue}{To mention a few, the authors in \cite{ZJWL-2018} suppose that all agents can reach dynamic consensus under the proposed self-triggered DMPC algorithm. The authors in \cite{MZLK-2019} focus on the trade-off between the communication cost and the control performance by explicitly including both quantities in the cost function. However, the controlled plant and the constraints considered in these results are too ideal that robustness discussion is not established. The influence of external disturbance should be taken into account since uncertainty always exists in practice. As a result, several attempts are made. For instance, in \cite{WZS-2021}, a min-max OCP is established under the framework of DMPC via the self-triggered generator to counteract the external interference. }}Nevertheless, the complexity of the constructed OCP conduces to a high degree of computing cost and even difficulty to solve. In general, the self-triggered DMPC algorithm should be provided with desirable robustness and low complexity, which motivates us to present this work.

Some challenges still remain for the self-triggered DMPC algorithm. For nonlinear MASs subject to external perturbation, the access of the self-triggered scheme makes the error between the actual state and the predicted state more difficult to estimate. It is not easy to ensure the robustness and coordination of the MAS under the self-triggered scheme. Additionally, an appropriate self-triggered generator should be searched to maintain superiority while reducing the deviation redundancy of the criteria. In particular, it is tough to guarantee the recursive feasibility of the algorithm and the stability of the system under the effect of the trigger mechanism. Hence, the presented work attempts to confront the issues and enrich the available research on this topic.

Inspired by the above discussion, this paper provides new insights into reducing the energy loss of the DMPC algorithm via the self-triggered generator and the prediction horizon strategy for the perturbed discrete-time MAS. The OCP for each agent is solved asynchronously with the respective prediction horizon. The main contributions are three aspects:
\begin{enumerate}
	\item[i)]
	A self-triggered DMPC mechanism is proposed for nonlinear MASs subject to bounded external perturbation.
	The cooperative control goal is achieved by balancing the control performance and energy loss.
	
	\item[ii)]
	By altering the form of the candidate Lyapunov function, the trigger prerequisite is derived in accordance with individual dynamics regardless of the neighbor information, leading to less reliance on communication.

	\item[iii)]
	The sufficient criteria are proposed, which is used to prove the recursive feasibility of proposed algorithm and the stability of closed-loop system.
\end{enumerate}

\emph{Notations:}
$\mathbb{R}$ and $\mathbb{R}^n$ denote the set of real integers and n-dimensional Euclidean space.
$\mathbb{N}$ denotes the set of nonnegative integers. Let $[a,b] = \{a,a+1,\dots,b\}$, $a,b \in \mathbb{N}$.
For a given vector $x$, $\| x \| = \sqrt{x^Tx}$ denotes the Euclidean norm, $\| x \| _{P}= \sqrt{x^TPx}$ denotes the $P$-weighted norm.
For a given matrix $P$, $P>0$ means $P$ is positive definite and $P\geq0$ means $P$ is positive semi-definite. $\overline{\lambda}(P)$ and $\underline{\lambda}(P)$ denote the maximum and minimum eigenvalues, respectively.
The Minkowski sum and Minkowski difference of two convex sets are described as $\mathcal{A} \oplus \mathcal{B} \triangleq \{ a+b \;|\; a \in \mathcal{A}, b \in \mathcal{B} \}$, $\mathcal{A} \ominus \mathcal{B} \triangleq \{ a \;|\; \forall b \in \mathcal{B},  a+b \in \mathcal{A} \}$.
$\mathcal{B}_{P}\left(r\right) = \{x \in \mathbb{R}^n: \left\| x \right\|^2_{P} \leq r^2 \}$ is an ellipsoid with weighted matrix $P$.

\begin{figure*}
	\centering
	\includegraphics[width=15cm]{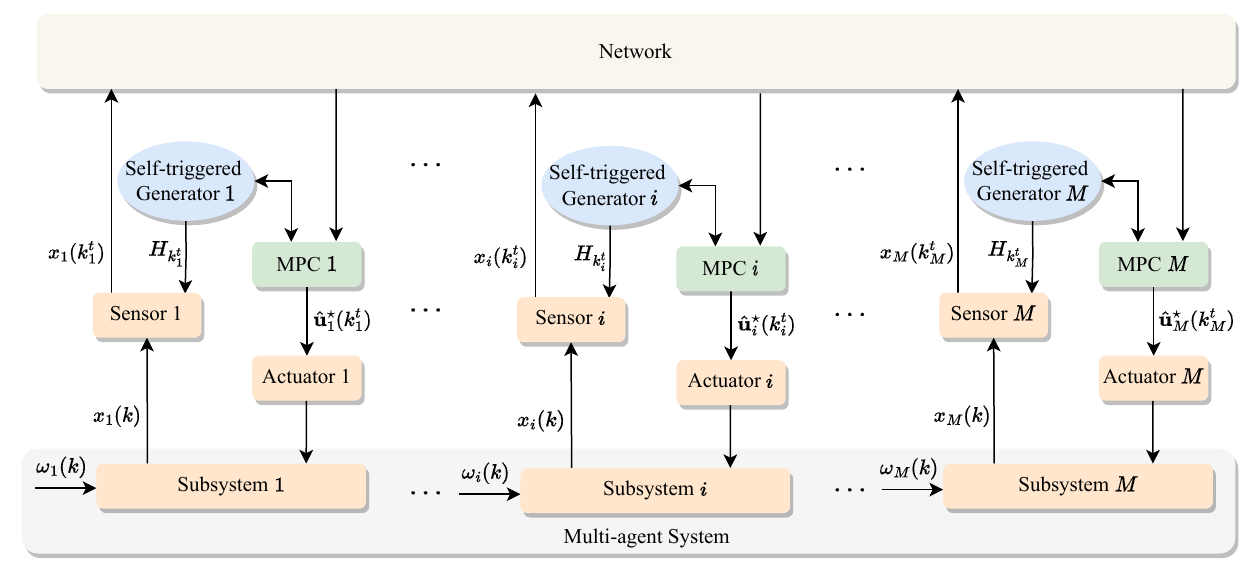}
	\caption{The diagram of DMPC algorithm with self-triggered coordinator.}
	\label{Fig: diagram}
\end{figure*}

\section{Preliminaries and Problem Statements}
\subsection{System Definition}
We consider a group of discrete-time nonlinear agents. {\textcolor{blue}{The communication topology of the  MAS is described by a directed graph $\mathcal{G} = (\mathcal{M}, \mathcal{E})$, where $\mathcal{M} = \{1,...,M\}$ is the node set and  $\mathcal{E} \in \mathcal{M} \times \mathcal{M}$ is the edge set. The existence of the edge $(i, j) \in \mathcal{E}$ indicates that agent $j$ can obtain the information from $i$th agent. Agent $j$ is called a neighbor of $i$ if $(i, j) \in \mathcal{E}$. And the set of neighbors of agent $i$ is denoted by $\mathcal{N}_i = \{ j\in\mathcal{M}:(j,i) \in \mathcal{E} \}$. In this paper, without losing generality, we suppose the graph $\mathcal{G}$ is strongly connected such that for each pair of agent $(i, j)$, there exists at least one path from $i$ to $j$ \cite{BJ-2014}.}} Agent $i$,  $i \in \mathcal{M}$, accords with the following description:
\begin{eqnarray}\label{subsystem}
  x_i (k+1) = f_i (x_i(k), u_i(k) ) + \omega_i(k),
\end{eqnarray}
where $x_i(k) \in \mathbb{R}^{n_i}$ is the real state of agent $i$ at time step $k$; $u_i(k) \in \mathbb{R}^{m_i}$ is the control input; $\omega_i(k) \in \mathbb{R}^{n_i}$ is the external disturbance. Suppose the subsystem $i$ is subject to hard constraints:
\begin{eqnarray*}
  x_i(k) \in \mathcal{X}_i, \quad u_i(k) \in \mathcal{U}_i, \quad \omega_i \in \mathcal{W}_i,
\end{eqnarray*}
where $\mathcal{X}_i \subset \mathbb{R}^{n_i}$ is closed, $\mathcal{U}_i \subset \mathbb{R}^{m_i}$ is compact, and $\mathcal{W}_i = \{ w_i(k) \in \mathbb{R}^{n_i} : \left\| \omega_i(k) \right\| \leq \eta_i \}$.

\begin{assumption}\label{Assumption: Lipschitz}
	\textcolor{blue}{The function $f_i: \mathbb{R}^{n_i} \times \mathbb{R}^{m_i} \mapsto \mathbb{R}^{n_i}$, with $f(\mathbf{0}, \mathbf{0}) = \mathbf{0}$ is twice differentiable and locally Lipschitz continuous in the first argument $x_i$. Let $g_i(x_i(k), u_i(k)) = f_i(x_i(k), u_i(k)) - x_i(k)$. There exists positive Lipschitz constant $L_i$ such that
		$||g_i(x_1, u) - g_i(x_2, u)|| \leq L_i ||x_1 - x_2||$ for any $x_1,x_2 \in \mathcal{X}_i$, $u_i \in \mathcal{U}_i$. If $u_i(k) = K_i x_i(k)$, the relevant Lipschitz constant is $L_{ri}$.}
\end{assumption}

In the subsequent design procedure, a local auxiliary controller is applied to restrain the system state. To this end, the linearized model of the above nominal subsystem at the origin is constructed based on the Jacobin linearization technique, which is described as follows:
\begin{eqnarray}\label{linearized subsystem}
  x_i(k+1) = A_i x_i\left(k\right) + B_i u_i(k) + \omega_i(k),
\end{eqnarray}
with $A_i = \frac{\partial f_i}{\partial x_i}|_{(0,0)}$, and $B_i = \frac{\partial f_i}{\partial u_i}|_{(0,0)}$.

\begin{assumption}\label{Assumption: Schur}
  Suppose the matrix pair $(A_i, B_i)$ is stabilizable for each linearized subsystem (\ref{linearized subsystem}), i.e., there exists a state feedback matrix $K_i$ such that $A_i + B_i K_i$ is Schur.
\end{assumption}

\subsection{Problem Formulation}
The primary control objective is to achieve the consensus of each subsystem and satisfy the constraints. Above all, we define the nominal system for subsystem $i$:
\begin{eqnarray}\label{nominal subsystem}
    \hat{x}_i(k+1) = f_i(\hat{x}_i(k), \hat{u}_i(k)),
\end{eqnarray}
where $\hat{x}_i(k) \in \mathbb{R}^{n_i}$ and $\hat{u}_i(k) \in \mathbb{R}^{m_i}$ are the nominal state and input, respectively.

As described, for each subsystem, the concrete dual-mode DMPC strategy can be presented as follows:
\begin{itemize}
  \item[1)]
  Before the state is steered into the terminal set $\mathcal{X}^r_i$, the optimization problem is implemented at each trigger instant. Let $k^t_i \in \mathbb{N}$ denotes the $t$-th trigger instant of agent $i$, $N_{k^t_i}$ denotes the prediction horizon of agent $i$ at $k^t_i$ calculated in advance. The initial prescribed prediction horizon is $N_0$. Suppose the  optimal control sequence and the corresponding state sequence for agent $i$ at $k^t_i$ are $\hat{\mathbf{u}}^{\star}_i(k^t_i) = \{ \hat{u}^{\star}_i (k^t_i | k^t_i), \hat{u}^{\star}_i (k^t_i +1 | k^t_i), ..., \hat{u}^{\star}_i (k^t_i + N_{k^t_i} -1 | k^t_i ) \}$ and $\hat{\mathbf{x}}^{\star}_i(k^t_i) = \{  \hat{x}^{\star}_i (k^t_i +1 | k^t_i ), ..., \hat{x}^{\star}_i (k^t_i + N_{k^t_i}| k^t_i) \}$, then the actual control input is
  \begin{eqnarray}\label{open phase control}
    u_i^{\text{mpc}}(k) = \hat{u}^{\star}_i(k^t_i +l | k^t_i), \quad l \in [0, \,\, H_{k^t_i}-1],
  \end{eqnarray}
  for $i \in \mathcal{M}$. $H(k^t_i) : \mathbb{N} \mapsto \mathbb{N}$ is the inter-execution time determined at $k^t_i$, which is specified in Section \ref{Subsection: Generator}. We note the next trigger instant is determined as
  \begin{eqnarray}\label{next triggering time}
    k^{t+1}_i = k^{t}_i + H_{k^t_i}.
  \end{eqnarray}
  When the system state approaches the terminal region gradually, it is rational to shrink the prediction domain. The prediction horizon of the next OCP is adjusted as
  \begin{eqnarray}\label{next prediction horizon}
    N_{k^{t+1}_i} = N_{k^t_i} - \bar{N}(H_{k^t_i}),
  \end{eqnarray}
  where $\bar{N}(H_{k^t_i}) : \mathbb{N} \mapsto \mathbb{N}$, will be specified in Section \ref{Subsection: Generator}, is the horizon shrinkage extent between the subsequent OCP.
  \item[2)]
  As the system enters the terminal set $\mathcal{X}^r_i$, the local controller switches to
  \begin{eqnarray}\label{terminal control}
    u_i^{\text{local}}(k) = K_i x_i(k),
  \end{eqnarray}
\end{itemize}

\section{Self-triggered DMPC Algorithm}
In this section, the concrete DMPC optimization problem solved at each trigger instant is demonstrated. Then the asynchronous information aggregation process is introduced. The distributed self-triggered scheduling mechanism and the adaptive prediction horizon updating strategy are elaborated. The control architecture of the overall system is depicted in Fig. \ref{Fig: diagram}. Moreover, the pseudocode of the proposed algorithm is induced to clarify the whole process.

\subsection{DMPC Optimization Problem}\label{Subsection: DMPC optimization problem}
At $k^t_i$, each agent $i\in\mathcal{M}$ constructs an estimated state sequence $\hat{\mathbf{x}}^{ref}_{\mathcal{N}_i}(k^t_i)$ based on the most recently transmitted information from neighbors, aiming to achieve consensus among all agents. The local cost function of the optimization problem is:
{\textcolor{blue}{
\begin{eqnarray*}
  J_i(k^t_i) &\triangleq& J_i (x_i(k^t_i), \hat{\mathbf{u}}_i(k^t_i), \hat{\mathbf{x}}^{ref}_{\mathcal{N}_i}(k^t_i), N_{k^t_i} ) \\
  &=& J^s_i (x_i(k^t_i), \hat{\mathbf{u}}_i(k^t_i), N_{k^t_i} )  + J^c_i (x_i(k^t_i), \hat{\mathbf{x}}^{ref}_{\mathcal{N}_i}(k^t_i), N_{k^t_i} ), \\
  &\triangleq& J^s_i(k^t_i) + J^c_i(k^t_i),
\end{eqnarray*}
where $J^s_i(k^t_i)$ is the egoistic cost which is contributory to itself, and $J^c_i(k^t_i)$ is the altruistic cost which is added to achieve consensus. Specifically,
\begin{eqnarray*}
	J^s_i (k^t_i)
	&=& {\textstyle \sum_{l=0}^{N_{k^t_i} - 1}} L_i(\hat{x}_i(k^t_i + l | k^t_i), \hat{x}_i(k^t_i + l | k^t_i)) \\
	&& +  V^f_i (\hat{x}_i(k^t_i + N_{k^t_i} | k^t_i),
\end{eqnarray*}
 where $V^f_i(x) = \| x \|^2_{P_i}$, $L_i(x,u)=\| x \|^2_{Q_i} + \| u \|^2_{R_i}$. The weighted matrices satisfy $Q_i\geq0$, $R_i>0$, $P_i\geq0$. And
\begin{eqnarray*}
	J^c_i (k^t_i) = {\textstyle \sum_{l=0}^{N_{k^t_i} - 1} \sum_{j \in \mathcal{N}_{i}} } \| \hat{x}_i(k^t_i + l | k^t_i) - \hat{x}^{ref}_j(k^t_i + l | k^t_i) \|^2_{Q_{ij}},
\end{eqnarray*}
where $Q_{ij}\geq0$. }} Since dual-mode DMPC scheme is adopted to explicitly guarantee the stability of the subsystem, the terminal cost and the terminal region are required to meet the following assumptions.
\begin{assumption}\label{Assumption: terminal}
	With respect to the nominal subsystem (\ref{nominal subsystem}), there exists a scalar $r_i > 0$, and a quadratic sublevel set $\mathcal{X}^r_i = \{ \hat{x}_i(k): V^f_i \left(\hat{x}_i(k)\right) \leq r^2_i \} \subseteq \mathcal{X}_i$, which is positively invariant, such that the following descriptions hold for $\forall \hat{x}_i(k) \in \mathcal{X}^r_i$: 
	
	1) $K_i\hat{x}_i(k) = \mathcal{U}_i$; 
	
	2) $\hat{x}_i(k+1) = f_i (\hat{x}_i(k), K_i \hat{x}_i(k) ) \in \mathcal{X}^r_i$; 
	
	3) $ V^f_i (\hat{x}_i(k+1)) - V^f_i(\hat{x}_i(k)) \leq - \| \hat{x}_i(k) \|^2_{\overline{Q}_i}$, where $\overline{Q}_i = Q_i + K^T_i R_i K_i$.
\end{assumption}

With respect to the original system dynamic (\ref{subsystem}) and nominal system dynamic (\ref{nominal subsystem}), the following lemma can be utilized to guarantee the upper bound to the state error:
\begin{lemma}[\cite{XDLX-2021}]\label{Lemma: Gronwall}
	Consider the subsystem (\ref{subsystem}) and the nominal subsystem (\ref{nominal subsystem}), at triggering instant $k^t_i$, they are controlled by the identical $l$-step open-loop control action sequence $\{ u_i(k^t_i|k^t_i), u_i(k^t_i+1|k^t_i), ..., u_i(k^t_i+l-1|k^t_i) \}$. Then, the predicted state error between the actual state $x_i(k^t_i+l)$ and the predicted state $\hat{x}_i(k^t_i+l|k^t_i)$ conforms to the following condition:
	\begin{eqnarray}\label{Lemma: xe norm bound 1}
		\left\| x_i(k^t_i+l) - \hat{x}_i(k^t_i+l|k^t_i) \right\|_{P_i} \leq \Gamma_{P_i}(l),
	\end{eqnarray}
	where $\Gamma_{\ast}(l)=\frac{\eta_i \bar{\lambda}(\sqrt{\ast})}{L_i} \left[(1+L_i)^l-1\right]$. Furthermore, we have
	\begin{eqnarray}\label{Lemma: xe norm bound 2}
		\Gamma_{P_i}(l) \leq l \eta_i \bar{\lambda}(\sqrt{P_i}) (1+L_i)^{l-1}. \quad
	\end{eqnarray}
\end{lemma}

For each nominal subsystem (\ref{nominal subsystem}), the optimization problem at triggered instant $k^t_i$ is formulated as follows:
\begin{eqnarray}\label{OCP problem}
  \mathcal{P}_i:
  \min_{\hat{\mathbf{u}}_i(k^t_i)} J_i (x_i(k^t_i), \hat{\mathbf{u}}_i(k^t_i), \hat{\mathbf{x}}^{ref}_{\mathcal{N}_i}(k^t_i), N_{k^t_i} ),
\end{eqnarray}
subject to
\begin{numcases}{}
	\hat{x}_i(k^t_i+l+1|k^t_i) & \nonumber \\
	\label{OCP constraint: agent i dynamic} & \hspace{-6.6cm}	 $=f_i (\hat{x}_i(k^t_i+l|k^t_i), \hat{u}_i(k^t_i+l|k^t_i)), \; l \in [0, N_{k^t_i}-1]$,  \\
	\label{OCP constraint: current state}
	 \hat{x}_i(k^t_i|k^t_i) = x_i(k^t_i), \\
	\label{OCP constraint: control constraint}
	 \hat{u}_i(k^t_i+l|k^t_i) \in \mathcal{U}_i, \; l \in [0, N_{k^t_i}-1], \\	
	\label{OCP constraint: state constraint}
	 \hat{x}_i(k^t_i+l|k^t_i) \in \mathcal{X}_i \ominus \mathcal{X}^e_i(l), \; l \in [1, N_{k^t_i} - 1], \\
	\label{OCP constraint: terminal constraint}
	 \hat{x}_i(k^t_i + N_{k^t_i}|k^t_i) \in \mathcal{X}^f_i, \\
	\label{OCP constraint: egoistic Lyapunov bounded}
	 J^s_i (x_i(k^t_i), \hat{\mathbf{u}}_i(k^t_i), N_{k^t_i} ) \leq \gamma_i(k^t_i),
\end{numcases}
where $\mathcal{X}^e_i(l) = \{x^e_i: \left\| x^e_i \right\|_{P_i} \leq l \eta_i \bar{\lambda}(\sqrt{P_i}) (1+L_i)^{l-1} \}$ is the robust tightened constraint set, $ \mathcal{X}^f_i = \{ \hat{x}_i: \left\| \hat{x}_i \right\|^2_{P_i} \leq f^2_i \}$ is the terminal set, and the constraint (\ref{OCP constraint: egoistic Lyapunov bounded}) is introduced to guarantee the stability of the system. 

At $k^t_i$, after solving $\mathcal{P}_i$, suppose the optimal control solution for agent $i$ at $k^t_i$ is given by $\hat{\mathbf{u}}^{\star}_i(k^t_i) = \{\hat{u}^{\star}_i(k^t_i | k^t_i), \hat{u}^{\star}_i(k^t_i+1 | k^t_i), \dots, \hat{u}^{\star}_i(k^t_i + N_{k^t_i}-1 | k^t_i) \}$. The optimal state sequence is $\hat{\mathbf{x}}^{\star}_i(k^t_i) = \{\hat{x}^{\star}_i(k^t_i | k^t_i), \hat{x}^{\star}_i(k^t_i+1 | k^t_i), \dots, \hat{x}^{\star}_i(k^t_i + N_{k^t_i}| k^t_i) \}$ with $\hat{x}^{\star}_i(k^t_i | k^t_i) = x_i(k^t_i)$. The concrete $\gamma_i(k^t_i)$ is determined in pursuance of the interval between the current trigger instant and the last  trigger instant, i.e.,

1) if $H_{k^{t-1}_i} = 1$:
\begin{eqnarray*}
	\gamma_i(k^t_i)
	&=& J^{s\star}_i(k^{t-1}_i) + \Upsilon(k^{t-1}_i)   \\
	&& - \left\| x_i(k^{t-1}_i) \right\|^2_{Q_i} - \left\| \hat{u}^{\star}_i(k^{t-1}_i | k^{t-1}_i) \right\|^2_{R_i},
\end{eqnarray*}
where $J^{s\star}_i(k^{t-1}_i) = J^s_i (x_i(k^{t-1}_i), \hat{\mathbf{u}}^{\star}_i(k^{t-1}_i), N_{k^{t-1}_i} )$ is composed of the optimal control sequence $\hat{\mathbf{u}}^{\star}_i(k^{t-1}_i)$ at last time step.

2) if $H_{k^{t-1}_i} > 1$:
\begin{eqnarray*}
	\gamma_i(k^t_i) 
	&=& \bar{J}^s_i(k^t_i-1) + \Lambda(H_{k^{t-1}_i}) \\
	&& - \left\| x_i( k^t_i - 1 ) \right\|^2_{Q_i} - \left\| \hat{u}^{\star}_i(k^t_i - 1 | k^{t-1}_i) \right\|^2_{R_i},
\end{eqnarray*}
where $\bar{J}^s_i(k^t_i-1) = J^s_i (x_i(k^t_i-1), \hat{\mathbf{u}}_i(k^t_i-1), N_{k^{t-1}_i} )$ consists of the candidate control sequence $\hat{\mathbf{u}}_i(k^t_i-1)$ which will be expatiated in Subsection \ref{Subsection: Recursive feasibility}.

\begin{remark}
	It can be seen that the actual state at the previous step needs to be measured to construct $J^{s\star}_i(k^{t-1}_i)$ or $J^s_i(k^t_i-1)$ of the constraint (\ref{OCP constraint: egoistic Lyapunov bounded}). This is a compromise in order to obtain a distributed self-triggering algorithm that does not rely on neighbor information.
\end{remark}

\begin{figure}
  \centering
  \includegraphics[width=8.4cm]{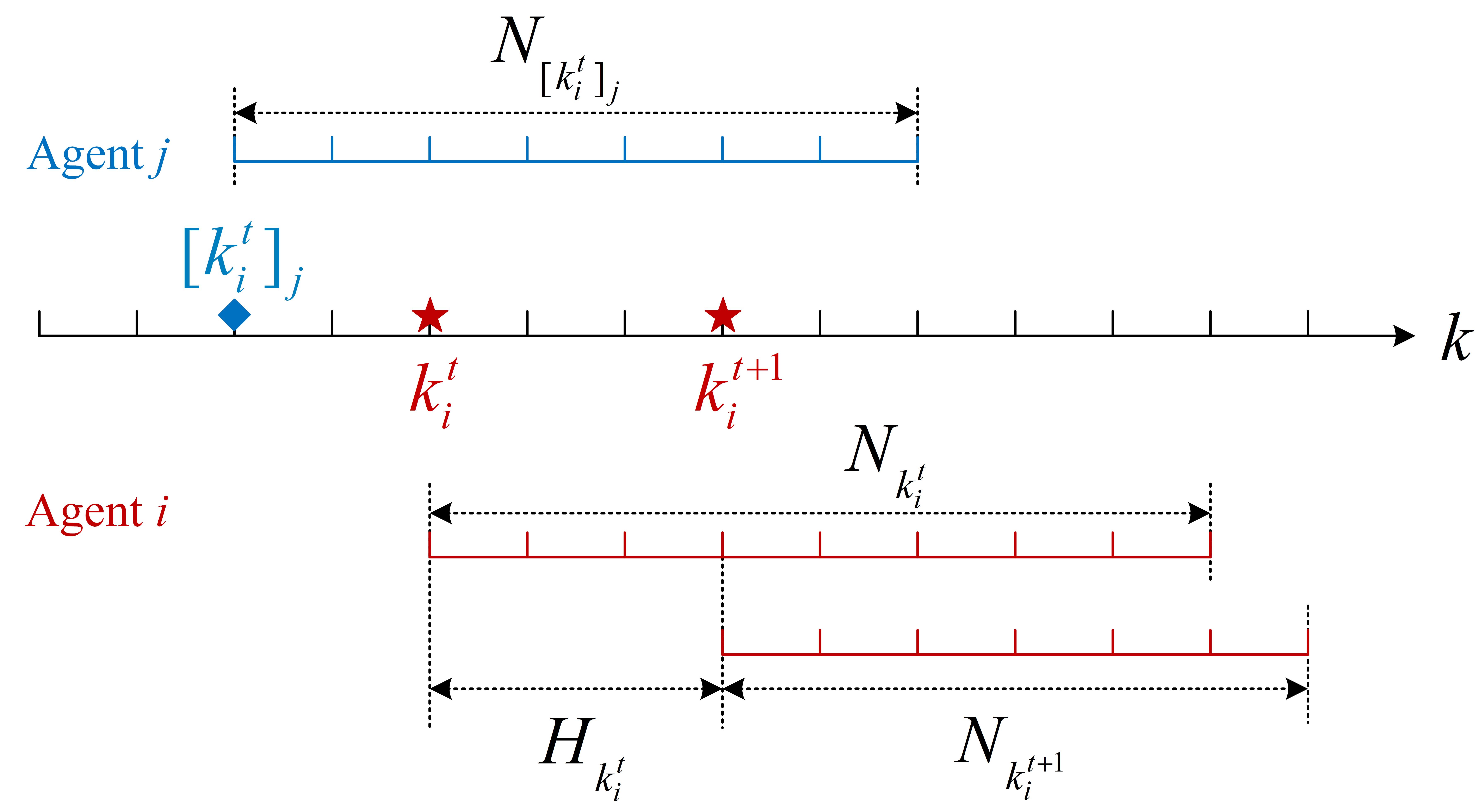}
  \caption{The relationships between $k^t_i$, $k^{t+1}_i$, $\left[k^t_i\right]_j$, $N_{k^t_i}$, $H_{k^t_i}$, and $N_{[k^t_i]_j}$.}
  \label{Fig: DMPC}
\end{figure}

\subsection{Asynchronous Information Aggregation}\label{}
The asynchronous communication scheme necessitates the estimated state trajectories of the neighbor agents.
We denote $[k^t_i]_j$ as the nearest triggering time of the neighbor agent $j$, $j\in\mathcal{N}_i$. The presumed state sequence $\hat{\mathbf{x}}^{ref}_j(k^t_i) = \{\hat{x}^{ref}_i(k^t_i | k^t_i), \hat{x}^{ref}_i(k^t_i+1 | k^t_i), \dots, \hat{x}^{ref}_i(k^t_i + N_{k^t_i} | k^t_i) \}$ can be classified into the following three situations:

\textit{Case 1:} $[k^t_i]_j + N_{[k^t_i]_j} \leq k^t_i$, which means that the neighbor $j$ has entered into the terminal set. The initial presumed state is given by
\begin{eqnarray*}
  \hat{x}^{ref}_j(k^t_i | k^t_i) =
  \hat{x}^{\star}_j\left([k^t_i]_j + N_{[k^t_i]_j} | [k^t_i]_j \right),
\end{eqnarray*}
and for $l \in [0, N_{k^t_i} - 1]$, the presumed control is
\begin{eqnarray*}
  \hat{u}^{ref}_j(k^t_i+l| k^t_i) =
  K_j \hat{x}^{ref}_j(k^t_i+l| k^t_i).
\end{eqnarray*}

\textit{Case 2:} $k^t_i < [k^t_i]_j + N_{[k^t_i]_j} \leq k^t_i + N_{k^t_i}$, which means that the neighbor $j$ provides segmental state sequence. The initial presumed state is defined as
\begin{eqnarray*}
  \hat{x}^{ref}_j(k^t_i | k^t_i) =
  \hat{x}^{\star}_j\left(k^t_i | [k^t_i]_j \right),
\end{eqnarray*}
and the presumed control is piece-wise as
\begin{eqnarray*}
  && \hat{u}^{ref}_j(k^t_i+l| k^t_i) \\
  && =
  \begin{cases}
    \hat{u}^{\star}_j (k^t_i | [k^t_i]_j), & l \in [0, \; [k^t_i]_j + N_{[k^t_i]_j} - k^t_i ], \\
    K_j \hat{x}^{ref}_j (k^t_i+l| k^t_i), & l \in [ [k^t_i]_j + N_{[k^t_i]_j} - k^t_i, \; N_{k^t_i} - 1 ].
  \end{cases}
\end{eqnarray*}

\textit{Case 3:} $[k^t_i]_j + N_{[k^t_i]_j} > k^t_i + N_{k^t_i}$, which means that the length of the state sequence sent by neighbor $j$ exceeds the length required by agent  $i$. The initial presumed state is
\begin{eqnarray*}
  \hat{x}^{ref}_j \left(k^t_i | k^t_i \right) =
  \hat{x}^{\star}_j \left(k^t_i | [k^t_i]_j \right),
\end{eqnarray*}
and for $l \in [0, N_{k^t_i}-1]$, the presumed control is
\begin{eqnarray*}
  \hat{u}^{ref}_j(k^t_i+l| k^t_i) =
  \hat{u}^{\star}_j\left(k^t_i | [k^t_i]_j \right).
\end{eqnarray*}

The specific state sequence of the neighbor agent $j$ can be established according to the following equation:
\begin{eqnarray}\label{agent j dynamic}\hspace*{-0.4cm}
	\hat{x}^{ref}_j(k^t_i+l+1|k^t_i) = f_j (\hat{x}^{ref}_j(k^t_i+l|k^t_i), \hat{u}^{ref}_j(k^t_i+l|k^t_i) ).
\end{eqnarray}

The relationship between the above primary variables is shown in Fig. \ref{Fig: DMPC}.

\subsection{Self-Triggered Generator}\label{Subsection: Generator}
The inter-event time $H_{k^t_i}$ introduced in (\ref{next triggering time}) is available by the following self-triggered generator:
\begin{eqnarray}\label{Generator}
	H_{k^t_i} = \min \left\{ H_{1}, \; H_{f_1}, \; H_{f_2}, \; H_{s} \right\},
\end{eqnarray}
with
\begin{eqnarray}
	\label{Generator: H=1}
	H_{1} &\!=\!&
	\begin{cases}
		1, & \Upsilon(k^t_i) \!>\! \sigma_i (\left\| x_i(k^t_i) \right\|^2_{Q_i} \!+\! \left\| \hat{u}^{\star}_i(k^t_i|k^t_i) \right\|^2_{R_i}), \\
		N_{k^t_i}, & else,
	\end{cases} \\
	\label{Generator: Feasibility1}
	H_{f_1} &\!=\!& \left\{ \sup H_{k^t_i}: \Phi_{P_i}(H_{k^t_i}) \leq r_i - f_i \right\}, \\
	\nonumber\label{Generator: Feasibility2}
	H_{f_2} &\!=\!& \! \left\{ \! \sup  H_{k^t_i}: \left( \sqrt{ 1 - \rho_i} \right)^{H_{k^t_i} - \bar{N}(H_{k^t_i})} \leq \frac{f_i}{f_i + \Phi_{P_i}(H_{k^t_i}) } \right\}, \\
	\\
	\label{Generator: Stability}
	\nonumber H_{s} &\!=\!& \Big\{ \sup H_{k^t_i}: \Lambda (H_{k^t_i}) \leq \sigma_i \Big ( \Theta_{Q_i}(H_{k^t_i} - 1)^2 \\
	&& \qquad \qquad \quad + \left\| \hat{u}^{\star}_i(k^{t}_i + H_{k^t_i} - 1 | k^t_i) \right\|^2_{R_i} \Big) \Big\},
\end{eqnarray}
where $0<\sigma_i<1$ is the triggering factor, $\rho_i = \underline{\lambda}(\overline{Q}_i) / \overline{\lambda}(P_i)$, $\Theta_{Q_i}(l) \triangleq \max \left\{\| \hat{x}^{\star}_i(k^t_i+l|k^t_i) \|_{Q_i} - \Gamma_{Q_i}(l), 0  \right\}$, $\Phi_{P_i}(H_{k^t_i}) = \Gamma_{P_i}(H_{k^t_i}) (1+L_i)^{N_{k^t_i} - H_{k^t_i}}$. To facilitate writing, we denote
\begin{eqnarray*}
	\Gamma_{\ast}(l) &=& \frac{\eta_i \bar{\lambda}(\sqrt{\ast})}{L_i} \left[(1+L_i)^l-1\right], \\
	\Xi_{\ast}(l) &=& \eta_i \bar{\lambda}(\sqrt{\ast}) (1+L_i)^l, \\
    \Psi_{\ast}(H_{k^t_i},l) &=& \eta_i \bar{\lambda}(\sqrt{\ast}) (1+L_i)^{N_{k^t_i} - H_{k^t_i}} (1+L_{ri})^{l}, \\
    \Omega_{\ast}(H_{k^t_i},l) &=& \Gamma_{\ast}(H_{k^t_i} - 1) (1+L_i)^{N_{k^t_i} - H_{k^t_i} + 1} (1+ L_{ri})^{l},
\end{eqnarray*}
and
\begin{eqnarray*}
	\Upsilon(k^t_i) &=& {\textstyle \sum_{l=0}^{N_{k^t_i} - 1}} \Big\{ \Xi_{Q_i}(l)^2 + 2 \Xi_{Q_i}(l) \left\| \hat{x}^{\star}_i(k^t_i + 1 + l|k^t_i) \right\|_{Q_i} \Big\}\\
	&& + \Xi_{P_i}(N_{k^t_i} - 1)^2 + 2 \Xi_{P_i}(N_{k^t_i} - 1) f_i, \\
	\Lambda(H_{k^t_i}) &=& \Lambda_2 (H_{k^t_i}) + \Lambda_3 (H_{k^t_i}) + \Lambda_4 (H_{k^t_i}), \\
	\Lambda_2 (H_{k^t_i}) &=& {\textstyle \sum_{l=0}^{N_{k^t_i} - H_{k^t_i} - 1} }
	\Big\{ \Xi_{Q_i}(l)^2 + 2 \Xi_{Q_i}(l) \times \\
	\nonumber&& \Big[\| \hat{x}^{\star}_i(k^t_i + H_{k^t_i} + l | k^t_i) \|_{Q_i} \\
	&&\qquad\qquad\quad\quad + \Gamma_{P_i}(H_{k^t_i} - 1) (1+L_i)^{l+1} \Big] \Big\}, \\
	\Lambda_3 (H_{k^t_i}) &=& {\textstyle \sum_{l = 0}^{H_{k^t_i} - 2}}
	\Big\{ \Psi_{\overline{Q}_i}(H_{k^t_i}, l)^2  + 2 \Psi_{\overline{Q}_i}(H_{k^t_i}, l)  \times \\
	\nonumber&& \left[ \| \hat{x}^r_i(k^t_i + N_{k^t_i} + l |k^t_i) \|_{\overline{Q}_i} + \Omega_{\overline{Q}_i}(H_{k^t_i}, l) \right] \Big\}, \\
	\Lambda_4 (H_{k^t_i}) &=& \Xi_{P_i}(H_{k^t_i}, H_{k^t_i} - 1)^2 + 2 \Xi_{P_i}(H_{k^t_i}, H_{k^t_i} - 1) \times \\
	&& \Big[\| \hat{x}^r_i(k^t_i + N_{k^t_i} + H_{k^t_i} - 1 | k^t_i) \|_{P_i} \\
	&&\qquad\qquad\qquad\quad + \Omega_{P_i}(H_{k^t_i} - 1, H_{k^t_i} - 1) \Big].
\end{eqnarray*}

Furthermore, the  prediction horizon shrinkage $\bar{N}(H_{k^t_i})$ presented in (\ref{next prediction horizon}) is decided as
\begin{eqnarray}\label{prediction horizon shrinkage}
	\bar{N}(H_{k^t_i}) = \min \{ H_{k^t_i} - 1, \; N_{k^t_i} - \hat{N}_{k^t_i}  \},
\end{eqnarray}
with $\hat{N}_{k^t_i} = \inf \{ l: \hat{x}^{\star}_i(k^t_i + l | k^t_i) \in \mathcal{X}^f_i, l \in [0, N_{k^t_i}-1 ] \}$.

In general,  the proposed self-triggering DMPC algorithm for each subsystem $i\in \mathcal{M}$ is detailed in Algorithm 1.
\begin{algorithm}
\caption{Self-Triggered DMPC with adaptive prediction horizon (ST-H-DMPC)}\label{alg1}
\begin{algorithmic}[1]
  \State
	\textbf{Require:} For subsystem $i \in \mathcal{M}$, the matrix $Q_i$, $R_i$, $Q_{ij}$, the initial prediction horizon $N_0$, the terminal weighted matrix $\mathcal{P}_i$ and the local state feedback gain $K_i$, the justified terminal region parameters $r_i$ and $f_i$.
	\State
	\textbf{Initialization.} Set $k=0$, $k^t_i = 0$,  $N_{k^t_i} = N_0$;
	
	\While {the control action is not stopped}
		\State
		Sample the system state $x_i(k)$;
		
		\While{$x_i(k) \notin \mathcal{X}^r_i$}
			\If{$k = k^{t}_i$}
				\State
				Update the initial state $\hat{x}_i(k^t_i|k^t_i)$ as in (\ref{OCP constraint: current state});
				\State
				Receive the available information of neighbor subsystems: $[k^t_i]_j$, $N_{[k^t_i]_j}$, $\hat{\mathbf{u}}^{\star}_j([k^t_i]_j)$, $\hat{\mathbf{x}}^{\star}_j([k^t_i]_j)$;
				\State
				Construct $\hat{x}^{ref}_j\left(k^t_i+l|k^t_i\right)$ as in (\ref{agent j dynamic});
				\State
				Solve $\mathcal{P}_i$ to obtain $\hat{\mathbf{u}}^{\star}_i(k^t_i)$ and $\hat{\mathbf{x}}^{\star}_i(k^t_i)$;
				\State
				Determine $H_{k^t_i}$ and $\bar{N}(H_{k^t_i})$  as in (\ref{Generator}), (\ref{prediction horizon shrinkage});
				\State
				Send $k^t_i$, $N_{k^t_i}$, $\hat{\mathbf{u}}^{\star}_i(k^t_i)$, $\hat{\mathbf{x}}^{\star}_i(k^t_i)$ to the neighbors;
			\Else
				\State
				Apply control action $u_i^{\text{mpc}}(k)$ as in (\ref{open phase control});
			\EndIf
			\State
			$k^t_i \gets k^t_i + H_{k^t_i}$, $k \gets k+1$;
		\EndWhile
		
		\State
		Apply control action $u_i^{\text{local}}(k)$ as in (\ref{terminal control}).
	\EndWhile
\end{algorithmic}
\end{algorithm}

\begin{remark}
	The condition of ensuring $H_{k^t_i} \geq 1$ is implied in the trigger tuner (\ref{Generator}), that is, $H_{f_1}$ and $H_{f_2}$ are satisfied by default for $H_{k^t_i}=1$. The analogous prerequisite for this can be found in Theorem 6 of \cite{XDLX-2021}.
\end{remark}

\begin{remark}
	In Algorithm 1, when $k=0$, the estimated
	 neighbor information in line 7 are specified as $\hat{x}^{ref}_j\left(k^t_i+l|k^t_i\right) = x_j\left(k^t_i\right), l \in [0, N_0 - 1 ]$, and the constraint (\ref{OCP constraint: egoistic Lyapunov bounded}) is excluded.
\end{remark}

\section{Analysis}
In this section, we proceed to analyze the recursive feasibility and stability of the proposed self-triggered DMPC algorithm.

\subsection{Recursive Feasibility}\label{Subsection: Recursive feasibility}
To ensure the effectiveness of the algorithm, we have to guarantee $\mathcal{P}_i$ has at least a feasible solution for each agent at each triggering instant. We define a candidate control sequence
\begin{eqnarray}\label{feasible control sequence}
  \nonumber\hspace{-1cm}&& \hat{\mathbf{u}}_i(k^{t+1}_i) =
  \{\hat{u}_i(k^{t+1}_i | k^{t+1}_i), \hat{u}_i(k^{t+1}_i + 1 | k^{t+1}_i), \\
  \hspace{-1cm}&& \qquad\qquad\qquad\quad\quad \dots, \hat{u}_i(k^{t+1}_i + N_{k^{t+1}_i} - 1 | k^{t+1}_i) \},
\end{eqnarray}
where
\begin{eqnarray*}
	\hspace{-0.2cm}&&
		\hat{u}_i \left(k^{t+1}_i+l | k^{t+1}_i \right) \\
	\hspace{-0.6cm}&& \quad =
		\begin{cases}
			\hat{u}^{\star}_i\left(k^{t+1}_i + l | k^{t}_i \right), & l \in [0, N_{k^{t}_i} - H_{k^{t}_i} - 1], \\
			K_i \hat{x}_i \left(k^{t+1}_i + l | k^{t+1}_i \right), & l \in [N_{k^{t}_i} - H_{k^{t}_i}, N_{k^{t+1}_i} - 1].
		\end{cases}
\end{eqnarray*}

Then the associated candidate state sequence is
\begin{eqnarray}\label{feasible state sequence}
	\nonumber&& \hat{\mathbf{x}}_i(k^{t+1}_i) =
	\big\{\hat{x}_i(k^{t+1}_i | k^{t+1}_i), \hat{x}_i(k^{t+1}_i + 1 | k^{t+1}_i), \\
	&& \qquad\qquad\qquad\quad\qquad \dots, \hat{x}_i(k^{t+1}_i + N_{k^{t+1}_i} | k^{t+1}_i) \big\},
\end{eqnarray}
with $\hat{x}_i(k^{t+1}_i | k^{t+1}_i) = x_i(k^{t+1}_i)$. To be useful in Section \ref{Subsection: Stability}, we define an additional terminal sequence:
\begin{eqnarray}\label{terminal state sequence}
	\nonumber&& \hat{\mathbf{x}}^r_i(k^{t}_i) =
	\big\{
		\hat{x}^r_i(k^{t}_i + N_{k^t_i}| k^{t}_i),
		\hat{x}^r_i(k^{t}_i + N_{k^t_i} + 1 | k^{t}_i), \\
	&& \qquad\qquad\quad \dots,
		\hat{x}^r_i(k^{t+1}_i + N_{k^{t}_i} + H_{k^t_i} - 1 | k^{t+1}_i)
	\big\},
\end{eqnarray}
where $\hat{x}^r_i(k^{t}_i + N_{k^t_i}| k^{t}_i) = \hat{x}^{\star}_i(k^{t}_i + N_{k^t_i}| k^{t}_i)$, and for $l \in [1, H_{k^t_i} - 1]$: $\hat{x}^r_i(k^{t}_i + N_{k^t_i} + l | k^{t}_i) = f_i (\hat{x}^r_i(k^{t}_i + N_{k^t_i} + l - 1| k^{t}_i), K_i\hat{x}^r_i(k^{t}_i + N_{k^t_i} + l - 1| k^{t}_i))$. First, we give the following lemma to facilitate the subsequent derivation.

\begin{lemma}\label{Lemma: error norm between optimal and candidate}
	For a presumed trigger interval $H_{k^t_i} \in [1,...,N_{k^t_i}]$, the candidate state for $k^t_i + H_{k^t_i}$ instant and the optimal state for $k^t_i$ instant satisfy the following inequality for  $l \in [0, N_{k^t_i} - H_{k^t_i}]$:
	\begin{eqnarray*}
		&& \left\| \hat{x}_i(k^{t}_i + H_{k^t_i} + l | k^t_i + H_{k^t_i}) - \hat{x}^{\star}_i(k^{t}_i + H_{k^t_i} + l | k^t_i) \right\|_{P_i} \\
		&& \leq  \Gamma_{P_i}(H_{k^t_i}) (1+L_i)^l.
	\end{eqnarray*}
\end{lemma}
\begin{proof}
	It follows from (\ref{Lemma: xe norm bound 1}) that $\| x_i (k^{t}_i + H_{k^t_i}) - \hat{x}^{\star}_i (k^{t}_i + H_{k^t_i} | k^{t}_i ) \|_{P_i} \leq \Gamma_{P_i}(H_{k^t_i})$. From the special Gronwall inequality \cite{H-2009}, we have
	\begin{eqnarray*}
		&& \left\| \hat{x}_i(k^{t}_i + H_{k^t_i} + l | k^{t}_i + H_{k^t_i}) - \hat{x}^{\star}_i(k^{t}_i + H_{k^t_i} + l | k^t_i) \right\|_{P_i} \\
		&\leq&  \left\| \hat{x}_i (k^{t}_i + H_{k^t_i}| k^{t}_i + H_{k^t_i}) - \hat{x}^{\star}_i (k^{t}_i + H_{k^t_i} | k^{t}_i ) \right\|_{P_i} \\
		&\quad& \!+\!\! \sum_{\tau=0}^{l-1} L_i \! \left\| \hat{x}_i (k^{t}_i \!+\! H_{k^t_i} \!+\! \tau | k^{t}_i \!+ H_{k^t_i} ) - \hat{x}^{\star}_i \! (k^{t}_i \!+ H_{k^t_i} \!+\! \tau | k^{t}_i ) \right\|_{P_i} \\
		&\leq&  \Gamma_{P_i}(H_{k^t_i}) (1+L_i)^l.
	\end{eqnarray*}
\end{proof}

\begin{theorem}
  For subsystem $i$, suppose that $\mathcal{P}_i$ has a solution at the initial time $k^0_i$. Then by the proposed Algorithm \ref{alg1}, the $\mathcal{P}_i$ is recursively feasible if $\mathcal{X}^r_i \subseteq \mathcal{X}_i \ominus \mathcal{X}^e_i(N_0)$.
\end{theorem}

\begin{proof}
It is obvious that the candidate control sequence (\ref{feasible control sequence}) complies with the constraints (\ref{OCP constraint: agent i dynamic})-(\ref{OCP constraint: current state}). The condition (\ref{OCP constraint: control constraint}) is established based on Assumption \ref{Assumption: terminal}. We then prove the satisfaction of (\ref{OCP constraint: state constraint})-(\ref{OCP constraint: egoistic Lyapunov bounded}) successively.

Firstly, we have to prove $\hat{x}_i(k^{t+1}_i + N_{k^{t+1}_i} | k^{t+1}_i) \in \mathcal{X}^f_i$. For the prediction interval $[1, N_{k^t_i} - H_{k^t_i}]$, Lemma \ref{Lemma: error norm between optimal and candidate} presents the relationship between $\hat{x}_i (k^{t+1}_i+l | k^{t+1}_i )$ and $\hat{x}^{\star}_i (k^{t+1}_i+l | k^{t}_i )$. When $l = N_{k^t_i} - H_{k^t_i}$,
we have
\begin{eqnarray*}
	\left\| \hat{x}_i(k^t_i + N_{k^t_i} | k^{t+1}_i) - \hat{x}^{\star}_i( k^t_i + N_{k^t_i} | k^t_i) \right\|_{P_i} \leq \Phi_{P_i}(H_{k^t_i}).
\end{eqnarray*}

Following $\| \hat{x}^{\star}_i (k^{t}_i + N_{k^t_i} | k^{t}_i) \|_{P_i} \leq f_i$ and the self-triggered generator (\ref{Generator: Feasibility1}), we get
\begin{eqnarray}\hspace{-0.6cm}\label{Theorem proof: r terminal set}
   \left\| \hat{x}_i (k^{t}_i +  N_{k^t_i} | k^{t+1}_i ) \right\|_{P_i} \leq f_i + \Phi_{P_i}(H_{k^t_i}) \leq r_i,
\end{eqnarray}
which implies that $\hat{x}_i (k^{t}_i + N_{k^t_i} | k^{t+1}_i )$ is steered to the terminal region $\mathcal{X}^r_i$. As for the prediction interval $[N_{k^t_i} - H_{k^t_i} + 1, N_{k^{t+1}_i} - 1 ]$, upon using Assumption \ref{Assumption: terminal} yields
\begin{eqnarray*}
	&& \left\| \hat{x}_i (k^{t+1}_i + N_{k^{t+1}_i} | k^{t+1}_i ) \right\|^2_{P_i} - \left\| \hat{x}_i  (k^{t+1}_i + N_{k^{t+1}_i} - 1 | k^{t+1}_i ) \right\|^2_{P_i} \\
	&& \leq - \left\| \hat{x}_i (k^{t+1}_i + N_{k^{t+1}_i} - 1 | k^{t+1}_i ) \right\|^2_{Q_{Ki}} \\
	&& \leq - \rho_i \left\| \hat{x}_i (k^{t+1}_i + N_{k^{t+1}_i} - 1 | k^{t+1}_i ) \right\|^2_{P_i}.
\end{eqnarray*}

Furthermore, we have
\begin{eqnarray*}
    &&
    \left\| \hat{x}_i(k^{t+1}_i + N_{k^{t+1}_i} | k^{t+1}_i) \right\|^2_{P_i} \\
    && \leq
    \left(1-\rho_i\right) \left\| \hat{x}_i (k^{t+1}_i + N_{k^{t+1}_i} - 1 | k^{t+1}_i ) \right\|^2_{P_i} \\
    && \leq
    \left(1-\rho_i\right)^{H_{k^t_i} - \bar{N}(H_{k^t_i}) } \left\| \hat{x}_i (k^{t}_i + N_{k^t_i} | k^{t+1}_i) \right\|^2_{P_i}.
 \end{eqnarray*}

By following the former inequality relation in (\ref{Theorem proof: r terminal set}) with less conservativeness and the self-triggered condition (\ref{Generator: Feasibility2}), one has
\begin{eqnarray*}
    &&
    \left\| \hat{x}_i(k^{t+1}_i + N_{k^{t+1}_i} | k^{t+1}_i) \right\|^2_{P_i} \\
    && \leq
    \left(1-\rho_i\right)^{H_{k^t_i} - \bar{N}(H_{k^t_i}) } \left[ f_i + \Phi_{P_i}(H_{k^t_i}) \right]^2 \leq f^2_i,
\end{eqnarray*}
which means that the candidate $\hat{u}_i \left(k^{t+1}_i+l | k^{t+1}_i \right)$ renders satisfaction of the terminal constraint (\ref{OCP constraint: terminal constraint}).

Secondly, we turn to the establishment of $\hat{x}_i(k^{t+1}_i+l|k^{t+1}_i) \in \mathcal{X}_i \ominus \mathcal{X}^e_i(l)$. Since we have $\| \hat{x}_i (k^{t+1}_i+l | k^{t+1}_i) \|_{P_i} \leq \| \hat{x}^{\star}_i (k^{t+1}_i+l | k^{t}_i) \|_{P_i} + \Gamma_{P_i}(H_{k^t_i}) (1+L_i)^l$, for $l \in [1, N_{k^t_i} - H_{k^t_i}]$, the following description holds within this interval:
\begin{eqnarray*}
    \hat{x}_i \left(k^{t+1}_i+l | k^{t+1}_i \right) 
    &\in& \mathcal{X}_i \ominus \mathcal{X}^e_i(l + H_{k^t_i}) \\
    &\oplus& \mathcal{B}_{P_i} \left( \Gamma_{P_i}(H_{k^t_i}) (1+L_i)^l \right).
\end{eqnarray*}

Recall the relationship (\ref{Lemma: xe norm bound 2}), we have
\begin{eqnarray*}
	&& \mathcal{B}_{P_i} \left( \Gamma_{P_i}(H_{k^t_i}) (1+L_i)^l \right) \\ 
	&& \subseteq \mathcal{B}_{P_i} \left( H_{k^t_i} \eta_i \bar{\lambda}(\sqrt{P_i}) (1+L_i)^{l + H_{k^t_i} - 1} \right).
\end{eqnarray*}

From the definition of $\mathcal{X}^e_i(l + H_{k^t_i})$, we can conclude
\begin{eqnarray*}
	\hat{x}_i \left(k^{t+1}_i+l | k^{t+1}_i \right) \in \mathcal{X}_i \ominus \mathcal{X}^e_i(l).
\end{eqnarray*}

For the interval $[N_{k^t_i} - H_{k^t_i} + 1, N_{k^{t+1}_i} - 1]$, the state is constrained within the terminal domain on the basis of Assumption \ref{Assumption: terminal}. Furthermore, it follows that $\hat{x}_i(k^{t+1}_i+l|k^{t+1}_i) \in \mathcal{X}^r_i \subseteq \mathcal{X}_i \ominus \mathcal{X}^e_i(N_0) \subseteq \mathcal{X}_i \ominus \mathcal{X}^e_i(l)$.

Finally, we proceed to confirm the feasibility of the constraint $\bar{J}^s_i (k^{t+1}_i) \leq \gamma_i(k^{t+1}_i)$.
For $H_{k^t_i} = 1$, one has $N_{k^{t+1}_i} = N_{k^{t}_i}$ according to (\ref{prediction horizon shrinkage}). Then $\bar{J}^s_i(k^{t}_i + 1)$ can be calculated as 
\begin{eqnarray*}
	\bar{J}^s_i(k^{t}_i + 1)
	&=& \bar{J}^s_i (x_i(k^{t}_i + 1), \hat{\mathbf{u}}_i(k^{t}_i + 1), N_{k^{t}_i} ) \\
	&\leq&  J^{s\star}_i (k^{t}_i) + \Upsilon(k^{t}_i)   \\
	&& - \left\| x_i(k^{t}_i) \right\|^2_{Q_i} - \left\| \hat{u}^{\star}_i(k^{t}_i | k^{t}_i) \right\|^2_{R_i} \\
	&=& \gamma_i(k^t_i+1).
\end{eqnarray*}
For $H_{k^t_i} > 1$, we have
\begin{eqnarray*}
	\bar{J}^s_i(k^{t}_i + H_{k^t_i})
	&=& \bar{J}^s_i (x_i(k^{t}_i + H_{k^t_i}), \hat{\mathbf{u}}_i(k^{t}_i + H_{k^t_i}), N_{k^{t+1}_i} ) \\
	&\leq& \bar{J}^s_i (x_i(k^{t}_i + H_{k^t_i}), \hat{\mathbf{u}}_i(k^{t}_i + H_{k^t_i}), N_{k^{t}_i} ) \\
	&\leq&  \bar{J}^{s}_i (k^{t}_i + H_{k^t_i} - 1) + \Lambda(H_{k^{t}_i})   \\
	&& - \left\| x_i(k^{t}_i  + H_{k^t_i} - 1) \right\|^2_{Q_i} \\
	&& - \left\| \hat{u}^{\star}_i(k^{t}_i  + H_{k^t_i} - 1 | k^{t}_i) \right\|^2_{R_i} \\
	&=& \gamma_i(k^t_i + H_{k^t_i}).
\end{eqnarray*}

The above two inequalities will be specified in detail in (\ref{Theorem proof: DeltaJ final for H=1}) and (\ref{Theorem proof: DeltaJ final for H>1}) in the next subsection.
It shows that $J^s_i (k^{t+1}_i) \leq \gamma_i(k^{t+1}_i)$ is valid in both cases. This completes the proof.
\end{proof}

\subsection{Stability}\label{Subsection: Stability}
In this subsection, we show that the stability of the subsystem $i$ is guaranteed under the proposed self-triggered mechanism by considering the egoistic cost $J^s_i(k^t_i)$ as a Lyapunov function. 

\begin{remark}
	To achieve the trigger rule with stability assurance, we have to find a Lyapunov function that decreases over time. If we regard the total cost $J_i(k^t_i)$ as a Lyapunov candidate, the influences of the neighbor agents should be taken into account. However, at time instant $k^t_i$, we only have the information transmitted by the neighboring agents recently, which renders the stability-based self-triggered condition quite conservative. Therefore, the egoistic cost $J^s_i(k^t_i)$ is regarded as the Lyapunov function to achieve the balance between triggering conservativeness and solving performance. Since the objective function of the optimization problem $\mathcal{P}_i$ is the total cost $J_i(k^t_i)$ instead of the egoistic cost $J^s_i(k^t_i)$, $J^{s\star}_i(k^t_i) \leq \bar{J}^s_i(k^t_i)$ is not necessarily true. Therefore, the constraint (\ref{OCP constraint: egoistic Lyapunov bounded}) is introduced to guarantee the stability of each agent. Similar analysis can refer to \cite{KSD-2015}.
\end{remark}

Above all, we introduce two lemmas to facilitate the subsequent derivation.

\begin{lemma}\label{Lemma: error norm between two candidate}
	For a presumed trigger interval $H_{k^t_i}\in [2,...,N_{k^t_i}]$, the candidate states at time instants $k^t_i + H_{k^t_i}$ and $k^t_i + H_{k^t_i} - 1$ instant satisfy the following inequalities:

	$i)$ For  $l \in [0, N_{k^t_i} - H_{k^t_i}]$:
	\begin{eqnarray*}
		&& \left\| \hat{x}_i(k^{t}_i \!+ H_{k^t_i} \!+ l | k^t_i \!+ H_{k^t_i}) - \hat{x}_i(k^{t}_i \!+ H_{k^t_i} \!+ l | k^t_i \!+ H_{k^t_i} \!- 1) \right\|_{P_i} \\
		&& \leq \eta_i \bar{\lambda}(\sqrt{P_i}) (1+L_i)^l.
	\end{eqnarray*}

	$ii)$ For  $l \in [0, H_{k^t_i} - 1]$:
	\begin{eqnarray*}
		&& \left\| \hat{x}_i(k^{t}_i \!+ N_{k^t_i} \!+ l | k^t_i \!+ H_{k^t_i}) - \hat{x}_i(k^{t}_i \!+ N_{k^t_i} \!+ l | k^t_i \!+ H_{k^t_i} \!- 1) \right\|_{P_i} \\
		&& \leq \eta_i \bar{\lambda}(\sqrt{P_i}) (1+L_i)^{N_{k^t_i} - H_{k^t_i}} (1+L_{ri})^{l}.
	\end{eqnarray*}
\end{lemma}
\begin{proof}
	According to $\hat{x}_i (k^t_i + H_{k^t_i}| k^t_i + H_{k^t_i}) = x_i (k^t_i + H_{k^t_i}) = f_i(x_i(k^t_i + H_{k^t_i} - 1), \hat{u}^{\star}_i(k^t_i + H_{k^t_i} - 1 | k^t_i )) + \omega_i(k^t_i + H_{k^t_i} - 1) $, and $\hat{x}_i (k^t_i + H_{k^t_i} | k^t_i + H_{k^t_i} - 1 ) = f_i(x_i(k^t_i + H_{k^t_i} - 1),\hat{u}_i(k^t_i + H_{k^t_i} - 1 | k^t_i + H_{k^t_i} - 1))$, we can conclude that $\| \hat{x}_i (k^t_i + H_{k^t_i}| k^t_i + H_{k^t_i}) - \hat{x}_i (k^t_i + H_{k^t_i} | k^t_i + H_{k^t_i} - 1 ) \|_{P_i} \leq \Gamma_{P_i}(1) = \eta_i \bar{\lambda}(\sqrt{P_i}) $. Furthermore, the difference between the feasible states at two consecutive instants are normalized as follows, for $l \in [0, N_{k^t_i} - H_{k^t_i}]$,
	\begin{eqnarray*}
		&& \left\| \hat{x}_i(k^{t}_i \!+ H_{k^t_i} \!+ l | k^t_i \!+ H_{k^t_i}) - \hat{x}_i(k^{t}_i \!+ H_{k^t_i} \!+ l | k^t_i \!+ H_{k^t_i} \!- 1) \right\|_{P_i} \\
		&& \leq \left\| \hat{x}_i (k^t_i + H_{k^t_i}| k^t_i + H_{k^t_i}) - \hat{x}_i (k^t_i + H_{k^t_i} | k^t_i + H_{k^t_i} - 1 ) \right\|_{P_i} \\
		&& \quad + \sum_{\tau=0}^{l-1} L_i \big\| \hat{x}_i (k^t_i + H_{k^t_i} + \tau | k^t_i + H_{k^t_i} ) \\
		&& \qquad\qquad\quad - \hat{x}_i (k^t_i + H_{k^t_i} + \tau | k^{t}_i + H_{k^t_i} - 1 ) \big\|_{P_i} \\
		&& \leq \eta_i \bar{\lambda}(\sqrt{P_i}) (1+L_i)^l,
	\end{eqnarray*}
	For $l \in [0, H_{k^t_i} - 1]$,
	\begin{eqnarray*}
		&& \left\| \hat{x}_i(k^{t}_i \!+ N_{k^t_i} \!+ l | k^t_i \!+ H_{k^t_i}) - \hat{x}_i(k^{t}_i \!+ N_{k^t_i} \!+ l | k^t_i \!+ H_{k^t_i} \!- 1) \right\|_{P_i} \\
		&& \leq  \left\| \hat{x}_i (k^t_i + N_{k^t_i} | k^t_i + H_{k^t_i}) - \hat{x}_i (k^t_i + N_{k^t_i} | k^t_i + H_{k^t_i} - 1 ) \right\|_{P_i} \\
		&& \quad + \sum_{\tau=0}^{l - 1} L_{ri} \big\| \hat{x}_i (k^t_i + N_{k^t_i} + \tau | k^t_i + H_{k^t_i} ) \\
		&& \qquad\qquad\qquad - \hat{x}_i (k^t_i + N_{k^t_i} + \tau | k^{t}_i + H_{k^t_i} - 1 ) \big\|_{P_i} \\
		&& \leq \eta_i \bar{\lambda}(\sqrt{P_i}) (1+L_i)^{N_{k^t_i} - H_{k^t_i}} (1+L_{ri})^{l}.
	\end{eqnarray*}
\end{proof}

\begin{lemma}\label{Lemma: error norm between terminal and candidate}
	The error norm between the candidate state for $k^t_i + H_{k^t_i}$ instant and the constructed terminal state sequence (\ref{terminal state sequence}) is bounded by the following inequality, for  $l \in [0, H_{k^t_i} - 1]$:
	\begin{eqnarray*}
		&& \left\| \hat{x}_i(k^{t}_i + N_{k^t_i} + l | k^t_i + H_{k^t_i} - 1) - \hat{x}^r_i(k^{t}_i + N_{k^t_i} + l | k^t_i) \right\|_{P_i} \\
		&& \leq \Gamma_{P_i}(H_{k^t_i} - 1) (1+L_i)^{N_{k^t_i} - H_{k^t_i} + 1} (1+ L_{ri})^{l}.
	\end{eqnarray*}
\end{lemma}
\begin{proof}
	Recall Lemma \ref{Lemma: error norm between optimal and candidate}, when $l = N_{k^t_i} - H_{k^t_i}$, we have
	\begin{eqnarray*}
		&& \left\| \hat{x}_i(k^{t}_i + N_{k^t_i} | k^t_i + H_{k^t_i}) - \hat{x}^{\star}_i(k^{t}_i + N_{k^t_i} | k^t_i) \right\|_{P_i} \\
		&& \leq \Gamma_{P_i}(H_{k^t_i}) (1+L_i)^{N_{k^t_i} - H_{k^t_i}}.
	\end{eqnarray*}

	Since $\hat{x}^r_i (k^{t}_i + N_{k^t_i} | k^{t}_i) = \hat{x}^{\star}_i(k^{t}_i + N_{k^t_i}| k^{t}_i)$, replacing $H_{k^t_i}$ with $H_{k^t_i} - 1$ of the above inequality yields
	\begin{eqnarray*}
		&& \left\| \hat{x}_i (k^{t}_i + N_{k^t_i} | k^{t}_i + H_{k^t_i} - 1 ) - \hat{x}^r_i (k^{t}_i + N_{k^t_i} | k^{t}_i ) \right\|_{P_i} \\
		&&\leq \Gamma_{P_i}(H_{k^t_i} - 1) (1+L_i)^{N_{k^t_i} - H_{k^t_i} + 1}.
	\end{eqnarray*}

	Similarly, we get
	\begin{eqnarray*}
		&& \left\| \hat{x}_i(k^{t}_i + N_{k^t_i} + l | k^{t}_i + H_{k^t_i} - 1) - \hat{x}^r_i(k^{t}_i + N_{k^t_i} + l | k^t_i) \right\|_{P_i} \\
		&\leq&  \left\| \hat{x}_i (k^{t}_i + N_{k^t_i} | k^{t}_i + H_{k^t_i} - 1 ) - \hat{x}^r_i (k^{t}_i + N_{k^t_i} | k^{t}_i ) \right\|_{P_i} \\
		&\quad& + \sum_{\tau=0}^{l - 1} L_{ri} \big\| \hat{x}_i ( k^{t}_i + N_{k^t_i} + \tau | k^{t}_i + H_{k^t_i} - 1 ) \\
		&&\qquad\qquad\qquad\qquad\qquad - \hat{x}^r_i( k^{t}_i + N_{k^t_i} + \tau | k^{t}_i ) \big\|_{P_i} \\
		&\leq&  \Gamma_{P_i}(H_{k^t_i} - 1) (1+L_i)^{N_{k^t_i} - H_{k^t_i} + 1} (1+ L_{ri})^{l}.
	\end{eqnarray*}
\end{proof}

\begin{theorem}
  For subsystem $i$, suppose that the condition of Theorem 1 is satisfied and the external disturbance is bounded. Then based on the proposed Algorithm 1, the state of the subsystem dynamic (\ref{subsystem}) is steered to a terminal region $\mathcal{X}^r_i$ within a finite time and is ultimately bounded.
\end{theorem}

\begin{proof}
  Above all, we prove that the system state starting from $\mathcal{X}_i \setminus \mathcal{X}^r_i$ will enter $\mathcal{X}^r_i$ in finite steps by considering the egoistic cost function as a Lyapunov function. We neglect the effect of prediction domain contraction here, and define the following differential term:
\begin{eqnarray*}
	H_{k^t_i} = 1: &\hspace*{-0.8em}& \Delta J^s_i(k^t_i + 1) = \bar{J}^s_i \left( k^t_i + 1 \right) - J^{s\star}_i \left(k^t_i \right), \\
	H_{k^t_i} \geq 2: &\hspace*{-0.8em}& \Delta J^s_i(k^t_i + H_{k^t_i}) = \bar{J}^s_i( k^t_i + H_{k^t_i}) - \bar{J}^{s}_i ( k^t_i + H_{k^t_i}-1).
\end{eqnarray*}

In view of the above two situations, the relevant theoretical analysis is derived separately.

\textit{Case 1:} $H_{k^t_i} = 1$, i.e., $k^{t+1}_i = k^t_i + 1$: Substituting the corresponding $\hat{\mathbf{u}}_i(k^t_i + 1)$ in (\ref{feasible control sequence}) into $\Delta J^s_i(k^t_i + 1)$ yields that
\begin{eqnarray*}
	\Delta J^s_i(k^t_i + 1) = \Delta^1_{i,1} + \Delta^1_{i,2} + \Delta^1_{i,3},
\end{eqnarray*}
where
\begin{eqnarray*}
	\Delta^1_{i,1} &=&  - \left\| x_i(k^t_i) \right\|^2_{Q_i} - \left\| \hat{u}^{\star}_i(k^t_i|k^t_i) \right\|^2_{R_i}, \\
	\Delta^1_{i,2} &=& \!\! \sum_{l=0}^{N_{k^t_i} - 1} \! \big\{ \left\| \hat{x}_i(k^t_i \!+\! 1 \!+\! l | k^t_i \!+\! 1) \right\|^2_{Q_i} \!-\! \left\| \hat{x}^{\star}_i(k^t_i \!+\! 1 \!+\! l | k^t_i) \right\|^2_{Q_i} \big\}, \\
	\Delta^1_{i,3} &=& \left\| \hat{x}_i(k^t_i+N_{k^t_i}|k^t_i+1) \right\|^2_{Q_i} + \left\| \hat{u}_i(k^t_i+N_{k^t_i}|k^t_i+1) \right\|^2_{R_i} \\
	&& + \left\| \hat{x}_i(k^t_i+1+N_{k^t_i}|k^t_i+1) \right\|^2_{P_i} \!-\! \left\| \hat{x}_i(k^t_i \!+\! N_{k^t_i} |k^t_i \!+\! 1) \right\|^2_{P_i} \\
	&&  + \left\| \hat{x}_i(k^t_i + N_{k^t_i} | k^t_i + 1) \right\|^2_{P_i} - \left\| \hat{x}^{\star}_i(k^t_i+N_{k^t_i}|k^t_i) \right\|^2_{P_i}.
\end{eqnarray*}

In pursuance of Lemma \ref{Lemma: error norm between optimal and candidate}, we have
\begin{eqnarray*}
	\left\| \hat{x}_i(k^t_i+1+l|k^t_i+1) - \hat{x}^{\star}_i(k^t_i+1+l|k^t_i) \right\|_{Q_i} \leq  \Xi_{Q_i}(l),
\end{eqnarray*}
for $l \in [0, N_{k^t_i} - 1]$. Then it follows that
\begin{eqnarray}\label{Theorem proof: DeltaJ2 final for H=1}
	\nonumber \Delta^1_{i,2} &\leq& \sum_{l=0}^{N_{k^t_i} - 1} \Xi_{Q_i}(l)^2 + 2 \Xi_{Q_i}(l) \left\| \hat{x}^{\star}_i(k^t_i+1+l|k^t_i) \right\|_{Q_i}. \\
	&&
\end{eqnarray}
The condition $\hat{x}_i ( k^{t}_i + N_{k^{t}_i} | k^{t}_i + 1) \in \mathcal{X}^r_i$ is revealed in (\ref{Theorem proof: r terminal set}). According to Assumption \ref{Assumption: terminal}, we have
\begin{eqnarray}\label{Theorem proof: Terminal condition}
	\nonumber&& \left\| \hat{x}_i (k^{t}_i + N_{k^t_i} + 1 | k^t_i + 1 ) \right\|^2_{P_i} - \left\| \hat{x}_i ( k^{t}_i + N_{k^t_i} | k^t_i + 1 ) \right\|^2_{P_i} \\
	&& \leq
	- \left\| \hat{x}_i (k^{t}_i + N_{k^t_i} | k^{t}_i + 1 ) \right\|^2_{\overline{Q}_i}.
\end{eqnarray}
Therefore, we have
\begin{eqnarray*}
	\Delta^1_{i,3} &\leq& \left\| \hat{x}_i(k^t_i + N_{k^t_i} | k^t_i + 1) \right\|^2_{P_i} -  \left\| \hat{x}^{\star}_i(k^t_i + N_{k^t_i} | k^t_i) \right\|^2_{P_i}.
\end{eqnarray*}

Let $H_{k^t_i} = 1, l = N_{k^t_i} - 1$ in Lemma \ref{Lemma: error norm between optimal and candidate}, one have
\begin{eqnarray*}
	&& \left\| \hat{x}_i(k^{t}_i + N_{k^t_i} | k^t_i + 1) - \hat{x}^{\star}_i(k^{t}_i + N_{k^t_i} | k^t_i) \right\|_{P_i} \leq \Xi_{P_i}(N_{k^t_i} - 1),
\end{eqnarray*}
then we have
\begin{eqnarray}\label{Theorem proof: DeltaJ3 final for H=1}
	\Delta^1_{i,3} \leq \Xi_{P_i}(N_{k^t_i} - 1)^2 + 2 \Xi_{P_i}(N_{k^t_i} - 1) f_i.
\end{eqnarray}
Substituting (\ref{Theorem proof: DeltaJ2 final for H=1}) and (\ref{Theorem proof: DeltaJ3 final for H=1}) into $\Delta J^s_i(k^t_i + 1)$ and following the trigger principle (\ref{Generator: Stability}), it is straightforward to derive that
\begin{eqnarray}\label{Theorem proof: DeltaJ final for H=1}
	\nonumber&&\Delta J^s_i(k^t_i + 1) \\
	\nonumber&\leq& \Upsilon(k^t_i) - \left\| x_i(k^t_i) \right\|^2_{Q_i} - \left\| \hat{u}^{\star}_i(k^t_i|k^t_i) \right\|^2_{R_i} \\
	&\leq& (\sigma_i - 1) \left( \left\| x_i(k^t_i) \right\|^2_{Q_i} + \left\| \hat{u}^{\star}_i(k^t_i|k^t_i) \right\|^2_{R_i} \right) < 0.
\end{eqnarray}

\textit{Case 2:} $H_{k^t_i} > 1$, i.e., $k^{t+1}_i = k^t_i + H_{k^t_i}$: Similarly, we have
\begin{eqnarray*}
	\Delta J^s_i(k^t_i + H_{k^t_i}) = \Delta^H_{i,1} + \Delta^H_{i,2} +  \Delta^H_{i,3} + \Delta^H_{i,4},
\end{eqnarray*}
where
\begin{eqnarray*}
	\Delta^H_{i,1} &=&  - \left\| x_i(k^{t}_i + H_{k^t_i} - 1 ) \right\|^2_{Q_i} \!-\! \left\| \hat{u}^{\star}_i(k^{t}_i + H_{k^t_i} - 1|k^t_i) \right\|^2_{R_i}, \\
	\Delta^H_{i,2} &=& \sum_{l = 0}^{N_{k^t_i} - H_{k^t_i} - 1} \!\!\! \big\{ \left\| \hat{x}_i(k^{t}_i + H_{k^t_i} + l | k^{t}_i + H_{k^t_i}) \right\|^2_{Q_i} \\
	&& \qquad\qquad\;\; - \left\| \hat{x}_i(k^{t}_i + H_{k^t_i} + l | k^{t}_i + H_{k^t_i} \!-\! 1) \right\|^2_{Q_i} \big\}, \\
	\Delta^H_{i,3} &=& \sum_{l = 0}^{H_{k^t_i} - 2} \big\{ \left\| \hat{x}_i(k^{t}_i + N_{k^t_i} + l | k^{t}_i + H_{k^t_i}) \right\|^2_{\overline{Q}_i} \\
	&& \qquad\qquad\; - \left\| \hat{x}_i(k^{t}_i + N_{k^t_i} + l | k^{t}_i + H_{k^t_i} - 1) \right\|^2_{\overline{Q}_i} \big\}, \\
	\Delta^H_{i,4} &=& \left\| \hat{x}_i(k^{t}_i + H_{k^t_i} + N_{k^t_i} - 1 | k^{t}_i + H_{k^t_i}) \right\|^2_{\overline{Q}_i} \\
	&& + \left\| \hat{x}_i(k^{t}_i + H_{k^t_i} + N_{k^t_i} | k^{t}_i + H_{k^t_i}) \right\|^2_{P_i} \\
	&& - \left\| \hat{x}_i(k^{t}_i + H_{k^t_i} + N_{k^t_i} - 1 | k^{t}_i + H_{k^t_i}) \right\|^2_{P_i} \\
	&& + \left\| \hat{x}_i(k^{t}_i + H_{k^t_i} + N_{k^t_i} - 1 | k^{t}_i + H_{k^t_i}) \right\|^2_{P_i} \\
	&& - \left\| \hat{x}_i(k^{t}_i + H_{k^t_i} - 1 + N_{k^t_i} | k^{t}_i + H_{k^t_i} - 1) \right\|^2_{P_i}.
\end{eqnarray*}

It yields from Lemma \ref{Lemma: Gronwall} that
\begin{eqnarray*}
	\left\| \hat{x}^{\star}_i(k^t_i + l | k^t_i) - x_i(k^t_i + l ) \right\|_{Q_i} &\leq& \Gamma_{Q_i}(l), \\
	\Longrightarrow \left\| \hat{x}^{\star}_i(k^t_i + l | k^t_i) \right\|_{Q_i} - \left\| x_i(k^t_i + l) \right\|_{Q_i} &\leq& \Gamma_{Q_i}(l),
\end{eqnarray*}
Then we have
\begin{eqnarray*}
	\left\| x_i(k^t_i + l) \right\|^2_{Q_i} \geq \Theta_{Q_i}(l)^2,
\end{eqnarray*}
and
\begin{eqnarray}\label{Theorem proof: DeltaJ1 final for H>1}
	\Delta^H_{i,1} &\leq&  - \Theta_{Q_i}(H_{k^t_i} - 1)^2 - \left\| \hat{u}^{\star}_i(k^{t}_i + H_{k^t_i} - 1 | k^t_i) \right\|^2_{R_i}.
\end{eqnarray}

It follows from Lemma \ref{Lemma: error norm between two candidate} that
\begin{eqnarray*}
	\Delta^H_{i,2}
	&\leq& {\textstyle \sum_{l=0}^{N_{k^t_i} - H_{k^t_i} - 1}} \Big\{ \Xi_{Q_i}(l)^2 + 2 \Xi_{Q_i}(l) \times \\
	&& \left\| \hat{x}_i(k^{t}_i + H_{k^t_i} + l | k^{t}_i + H_{k^t_i} - 1) \right\|_{Q_i} \Big\},
\end{eqnarray*}
Recall Lemma \ref{Lemma: error norm between optimal and candidate}, the above inequality can be rewritten as
\begin{eqnarray}\label{Theorem proof: DeltaJ2 final for H>1}
	\Delta^H_{i,2} \leq \Lambda_2(H_{k^t_i}).
\end{eqnarray}

By Lemma \ref{Lemma: error norm between two candidate}, one has
\begin{eqnarray*}
	\Delta^H_{i,3}
	&\leq& {\textstyle\sum_{l = 0}^{H_{k^t_i} - 2}} \Psi_{\overline{Q}_i}(H_{k^t_i}, l)^2 + 2 \Psi_{\overline{Q}_i}(H_{k^t_i}, l) \\
	&& \times \left\| \hat{x}_i(k^t_i + N_{k^t_i} + l | k^t_i + H_{k^t_i} - 1)  \right\|_{\overline{Q}_i}.
\end{eqnarray*}

Then it can be observed from Lemma \ref{Lemma: error norm between terminal and candidate} that
\begin{eqnarray}\label{Theorem proof: DeltaJ3 final for H>1}
	\Delta^H_{i,3} \leq \Lambda_3(H_{k^t_i}).
\end{eqnarray}

Like (\ref{Theorem proof: Terminal condition}), for $H_{k^t_i} > 1$, we have
\begin{eqnarray*}
	&& \left\| \hat{x}_i(k^{t}_i + H_{k^t_i} + N_{k^t_i} | k^{t}_i + H_{k^t_i}) \right\|^2_{P_i} \\
	&& - \left\| \hat{x}_i(k^{t}_i + H_{k^t_i} + N_{k^t_i} - 1 | k^{t}_i + H_{k^t_i}) \right\|^2_{P_i} \\
	&\leq&
	- \left\| \hat{x}_i (k^{t}_i + H_{k^t_i} + N_{k^t_i} - 1 | k^{t}_i + H_{k^t_i}) \right\|^2_{\overline{Q}_i}.
\end{eqnarray*}

Recall Lemma \ref{Lemma: error norm between two candidate}, $\Delta^H_{i,4}$ can be converted to
\begin{eqnarray*}
	\Delta^H_{i,4} &\leq&  \left\| \hat{x}_i(k^t_i + N_{k^t_i} + H_{k^t_i} - 1 | k^t_i + H_{k^t_i}) \right\|^2_{P_i} \\
	&& - \left\| \hat{x}_i(k^t_i + N_{k^t_i}  + H_{k^t_i} - 1 | k^t_i + H_{k^t_i} - 1) \right\|^2_{P_i} \\
	&\leq& \Xi^r_{P_i}(H_{k^t_i}, H_{k^t_i} - 1)^2 + 2 \Xi^r_{P_i}(H_{k^t_i}, H_{k^t_i} - 1) \\
	&& \times \left\| \hat{x}_i(k^t_i+ H_{k^t_i} - 1 + N_{k^t_i}|k^t_i + H_{k^t_i} - 1) \right\|_{P_i}.
\end{eqnarray*}

It can be observed from Lemma \ref{Lemma: error norm between terminal and candidate} that
\begin{eqnarray*}
	&& \left\| \hat{x}_i (k^{t}_i + N_{k^t_i} + H_{k^t_i} - 1 | k^t_i + H_{k^t_i} - 1) \right\|_{P_i} \\
	&\leq& \left\| \hat{x}^r_i(k^{t}_i + N_{k^t_i} + H_{k^t_i} - 1 | k^t_i) \right\|_{P_i} + \Omega_{P_i}(H_{k^t_i}, H_{k^t_i} - 1),
\end{eqnarray*}
substituting the above inequality into $\Delta^H_{i,4}$ yields
\begin{eqnarray}\label{Theorem proof: DeltaJ4 final for H>1}
	\Delta^H_{i,4} \leq \Lambda_4(H_{k^t_i}).
\end{eqnarray}

Combining (\ref{Theorem proof: DeltaJ1 final for H>1})-(\ref{Theorem proof: DeltaJ4 final for H>1}), and reflect the trigger rule (\ref{Generator: Stability}), we ultimately have
\begin{eqnarray}\label{Theorem proof: DeltaJ final for H>1}
	 \nonumber&& \Delta J^s_i(k^t_i + H_{k^t_i}) \\
	 \nonumber&\leq& \Lambda(H_{k^{t}_i}) - \left\| x_i(k^{t}_i + H_{k^t_i} - 1 ) \right\|^2_{Q_i} \!-\! \left\| \hat{u}^{\star}_i(k^{t}_i + H_{k^t_i} - 1|k^t_i) \right\|^2_{R_i} \\
	 \nonumber&\leq& (\sigma_i-1) \Big[ \Theta_{Q_i}(H_{k^t_i} - 1)^2 + \left\| \hat{u}^{\star}_i(k^{t}_i + H_{k^t_i} - 1 | k^t_i) \right\|^2_{R_i} \Big] < 0. \\
	 &&
\end{eqnarray}

Since both cases can guarantee $\Delta J^s_i(k^t_i + H_{k^t_i}) < 0$, we can conclude from \cite{KSD-2015} and \cite{MM-1993} that the state trajectory of subsystem $i$ will be steered into $\mathcal{X}^r_i$ within a finite time interval and is ultimately bounded. Thus completes the proof.
\end{proof}


\section{Illustrative Example}
\textcolor{blue}{
In this section, we shall demonstrate the validity and merits of the proposed algorithm. The position and posture control of four nonholonomic agents is taken into account \cite{GH-2006}. The state of each agent is characterized by $\chi_i(k) = [x_i(k), y_i(k), \theta_i(k)]^T$, where $[x_i(k), y_i(k)]^T$ represents the position vector of points on the moving path, and $\theta_i(k)$ is the angle that the tangents of the path at these points make with $x$-axis. The agents' motion is controlled by $u_i(k) = [v_i(k), w_i(k)]^T$, which includes the linear velocity $v_i(k)$ and the rotational velocity $w_i(k)$. The discrete-time nominal kinematic model is given by
\begin{eqnarray*}
	\left\{\begin{array}{rll}
		x_i(k+1) &=&
		x_i(k) + T v_i(k) \cos{\theta_{i}(k)}, \\
     	y_i(k+1) &=&
		y_i(k) + T v_i(k) \sin{\theta_{i}(k)}, \\
		\theta_i(k+1) &=&
		\theta_i(k) + T w_i(k),
	\end{array}\right.
\end{eqnarray*}
where the sample period is $T = 0.5$ s.  The state constraint set and the input constraint set are defined as
$\mathcal{X}_i = \left\{ \chi_i(k) \in \mathbb{R}^3 : |x_i(k)| \leq 1, |y_i(k)| \leq 1, |\theta_i(k)| \leq \pi/2 \right\}$, $\mathcal{U}_i =\left\{ u_i(k) \in \mathbb{R}^2 : |v_i(k)| \leq 1, |w_i(k)| \leq 0.6 \right\}$.
The external perturb is characterized by $\eta_i = 0.0001$.
Then, the corresponding Lipschitz constant interpreted in Assumption \ref{Assumption: Lipschitz} is obtained as $L_i = 0.5$ and $L_{ri} = 1.8581$ (see \cite{KG-2002}, Lemma 3.2 and 3.3).  The weighted matrices $Q_i = 0.8I_3$, $R_i = 0.5I_2$, and $Q_{ij} = I_3$. According to Assumption \ref{Assumption: terminal}, the terminal weighted matrix and the controller gain are calculated as
\begin{eqnarray*}
	P_i &=&
	\begin{bmatrix}
		2.3823 & 1.2083 & 1.0634 \\
		1.2083 & 2.7725 & 1.2213 \\
		1.0634 & 1.2213 & 2.4061
	\end{bmatrix}, \\
	K_i &=&
	\begin{bmatrix}
		-1.3332 & -1.2582 & -1.1247 \\
		-0.6310 & -0.7197 & -0.8395
	\end{bmatrix}.
\end{eqnarray*}
}

\textcolor{blue}{
The neighbor relationship between this MAS is $\mathcal{N}_1 = \{4\}$, $\mathcal{N}_2 = \{1\}$, $\mathcal{N}_3 = \{2\}$, and $\mathcal{N}_4 = \{3\}$. As for the terminal set, we set $r_1 = r_3 = 0.056$, $r_2 = r_4 = 0.065$, $f_1 = f_3 = f_4 = 0.03$, $f_2 = 0.04$. The self-triggered parameters are $\sigma_1 = \sigma_3 = 0.9$, $\sigma_2 = 0.5$, $\sigma_4 = 0.85$. The initial states  are  given by $\chi_1(0) = [-0.5, 0.9, \pi/6]^T$, $\chi_2(0) = [0.2, -0.4, -\pi/3]^T$, $\chi_3(0) = [0.4, -0.6, -\pi/4]^T$, and $\chi_4(0) = [-0.8, 0.4, \pi/5]^T$. The simulation verification is divided into three aspects:
}

\begin{figure}
	\includegraphics[width=8.7cm]{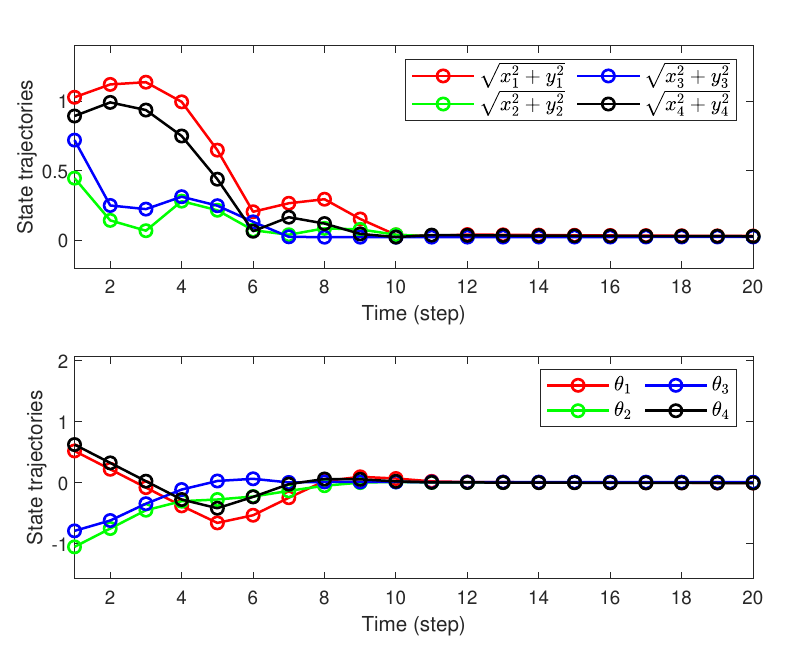}
	\caption{\textcolor{blue}{State trajectories of each agent under Algorithm 1.}}
	\label{Fig: state}
\end{figure}

\begin{figure}
	\centering
	\includegraphics[width=8.7cm]{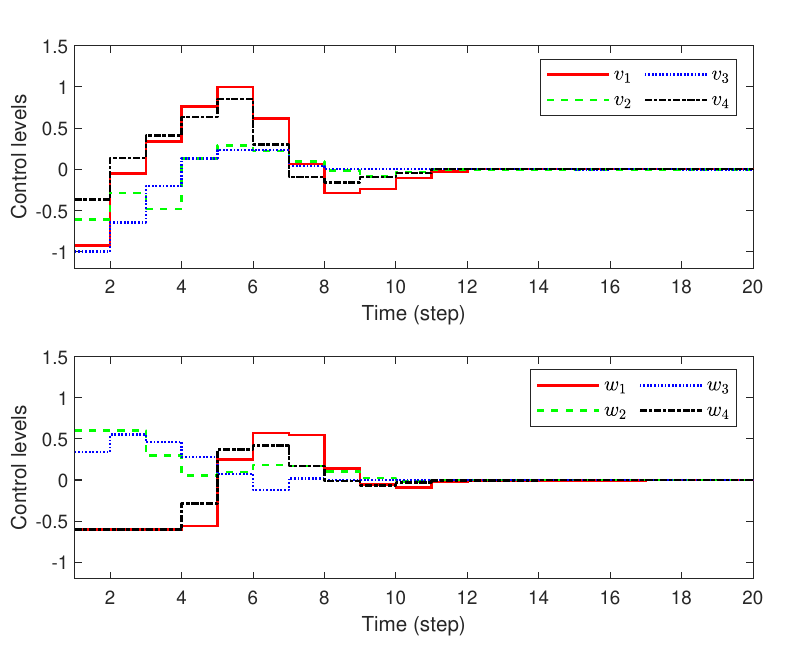}
	\caption{\textcolor{blue}{Control levels of each agent under Algorithm 1.}}
	\label{Fig: control}
\end{figure}

\begin{figure}
	\centering
	\includegraphics[width=8.7cm]{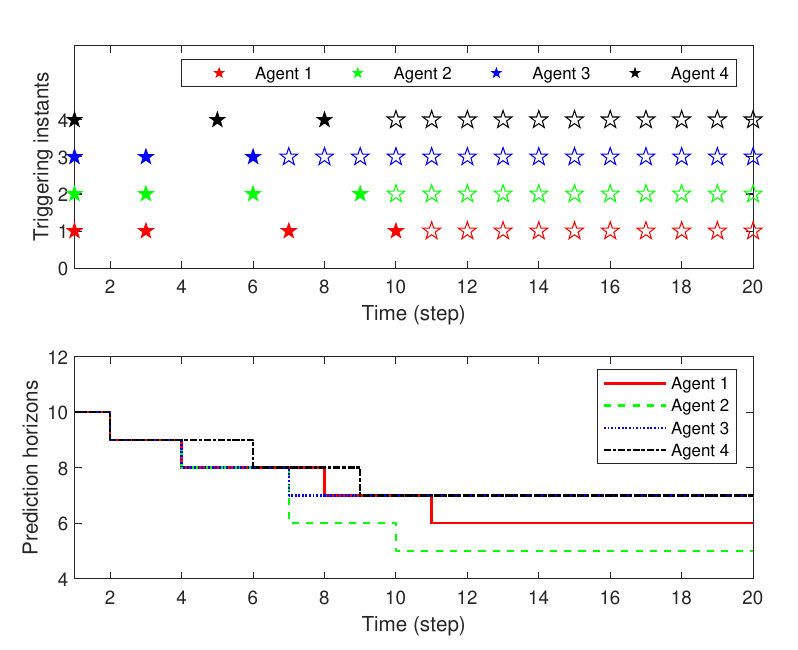}
	\caption{\textcolor{blue}{Triggering instants and prediction horizons at each triggered instant of each agent under Algorithm 1.}}
	\label{Fig: trigger and horizon}
\end{figure}

\begin{figure}
	\centering
	\includegraphics[width=8.7cm]{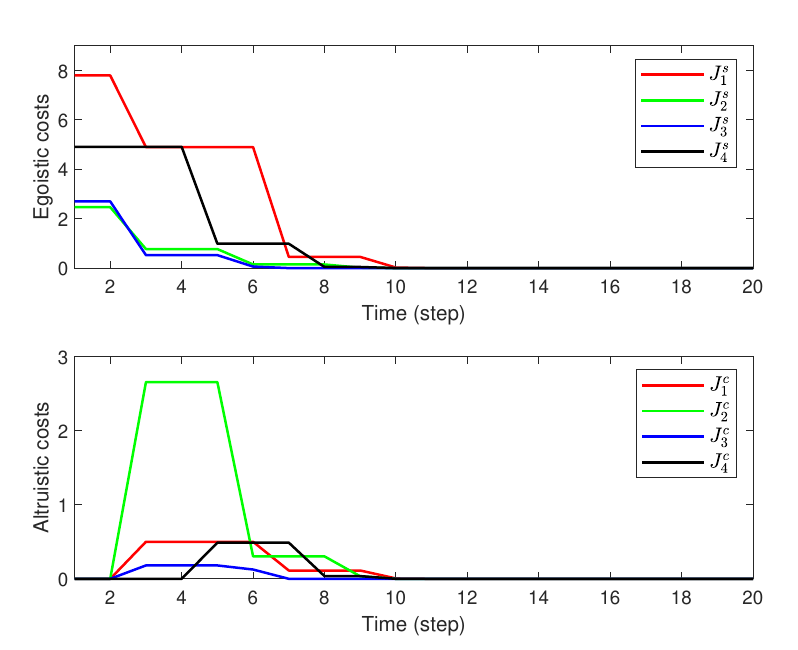}
	\caption{\textcolor{blue}{Cost levels of the egoistic cost $J^s_i$ and the altruistic cost $J^c_i$ for each agent under Algorithm 1.}}
	\label{Fig: cost}
\end{figure}

\textcolor{blue}{
(i) Effectiveness of the proposed coordination algorithm.
}

\textcolor{blue}{
The state trajectories and control levels of four subsystems are depicted in Fig. \ref{Fig: state} and Fig. \ref{Fig: control}.
The triggering instants of each agent are shown in Fig. \ref{Fig: trigger and horizon}.
The instants to solve $\mathcal{P}_i$ are emphasized by solid marks. The instants at which the controllers are switched into terminal controllers are opposite.
It shows that the corresponding prediction horizons are shrunk with the triggering instants under the proposed Algorithm 1. 
As can be seen from Fig. \ref{Fig: cost}, the Lyapunov function $J^s_i(k^t_j)$ is decreasing over the triggered instants. Each agent carries out collaborative optimization according to its own adjustment goal and the consensus goal of the distributed system.
}

\begin{figure}
	\centering
	\includegraphics[width=8.7cm]{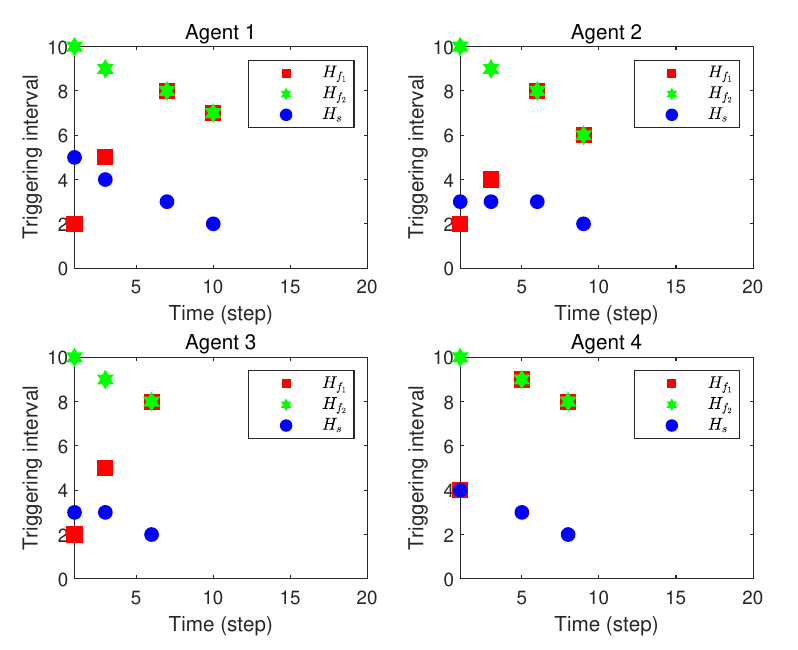}
	\caption{\textcolor{blue}{Self-triggered generator evolution of each agent under Algorithm 1.}}
	\label{Fig: generator}
\end{figure}

\textcolor{blue}{
(ii) Evolution of the proposed self-triggered generator.
}

\textcolor{blue}{
To clearly show the function of the proposed self-triggering mechanism, we show the value evolution of the generator at each triggering instant in Fig. \ref{Fig: generator} (only $H_{f_1}$, $H_{f_2}$ and $H_s$ are listed because $H_1$ is not activated in this example). It can be seen that the feasibility-based generator plays a decisive role in the first instant. According to Theorem 1, the proposed algorithm is recursively feasible, therefore the latter trigger interval is mainly determined by the generator based on stability.
}

\textcolor{blue}{
(iii) Comparison with the analogous algorithm.
}

\begin{figure}
	\centering
	\includegraphics[width=8.7cm]{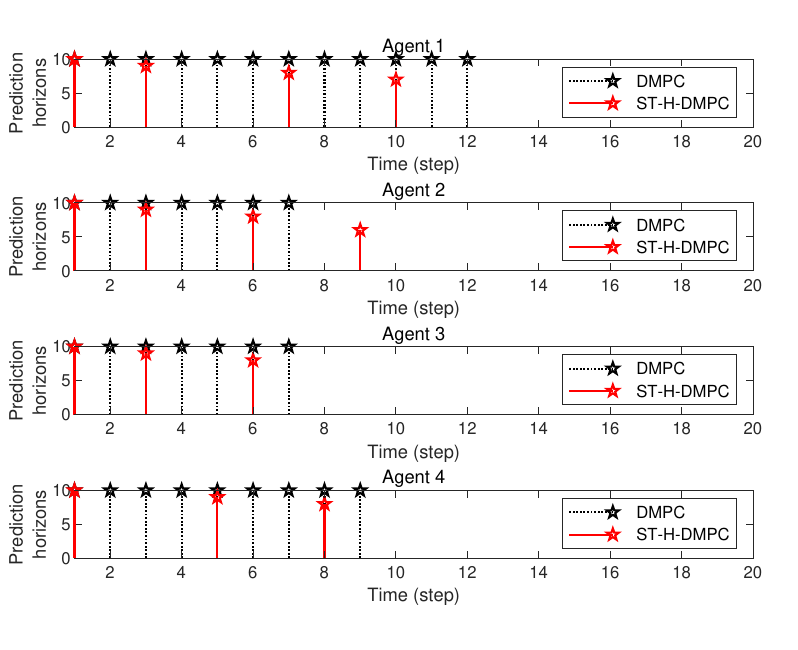}
	\caption{\textcolor{blue}{Triggering instants and the corresponding prediction horizons for each agent under DMPC algorithm and the proposed ST-H-DMPC algorithm.}}
	\label{Fig: with DMPC}
\end{figure}

\setlength{\tabcolsep}{0.5em}{
	\begin{table}[h]
		\renewcommand{\arraystretch}{1.5}
		\caption{\textcolor{blue}{The solving time (ms) of each agent during simulation time under four feasible algorithms.}}
		\label{Table: solving time}
		\centering
		\begin{tabular}{ccccc}
			\toprule
			& DMPC & H-DMPC & ST-DMPC & ST-H-DMPC \\
			\midrule
			Agent 1 & 404.63 & 337.38 & 142.54 &  91.32 \\
			Agent 2 & 180.07 & 161.56 & 122.95 &  91.79 \\
			Agent 3 & 147.84 & 147.72 &  72.86 &  80.18 \\
			Agent 4 & 221.42 & 208.19 &  78.81 &  71.42 \\
			Total   & 953.96 & 854.85 & 417.16 & 334.71 \\
			\bottomrule
		\end{tabular}
	\end{table}
}

\textcolor{blue}{
Furthermore, we compare the proposed self-triggered DMPC algorithm with adaptive prediction horizon (ST-H-DMPC) respectively with the traditional DMPC algorithm with fixed prediction horizon (DMPC), the traditional DMPC algorithm with adaptive prediction horizon (H-DMPC), the self-triggered DMPC algorithm with fixed prediction horizon (ST-DMPC).
As it can be seen in Fig. \ref{Fig: with DMPC}, the proposed ST-H-DMPC algorithm significantly reduces the number of solutions to the optimization problem.
On the other hand, the concrete solution time during the simulation process time is shown in Table. \ref{Table: solving time}. 
The solving time observed in ST-H-DMPC algorithm are far below those observed by other algorithms. Our proposed algorithm obviously reduces the calculation time.
}

\textcolor{blue}{
According to the above simulation results, we conclude that the proposed self-triggered DMPC algorithm with adaptive prediction horizon can not only achieve the coordination of the MAS but also reduce the computational burden effectively without persistently monitoring the system.
}

\section{Conclusions}
In this paper, we propose a self-triggered DMPC mechanism for perturbed discrete-time nonlinear MASs, where an upper bound constraint on partial cost is introduced into the optimization problem.
The desirable control performance and reduced energy consumption can be achieved through the distributed self-triggered generator and the prediction horizon shrinkage strategy collectively.
In view of the theoretical guarantee of the algorithm, we present a rigorous analysis of the recursive feasibility criteria and stability arguments.
The merits of the proposed algorithm are explicitly illustrated through the simulation results.

The trigger principle is deduced through the integrated use of the maximum norm on the error between the predicted state and the actual state and the reciprocal triangular relation.
A potential future work will focus on reducing the resulting conservativeness of the self-triggered criteria.

\nocite{*}
\bibliography{ref}%

\begin{thebibliography}{27}
\providecommand{\natexlab}[1]{#1}
\providecommand{\url}[1]{\texttt{#1}}
\providecommand{\urlprefix}{}

\bibitem{SVRWP-2010}
B.~T. Stewart et~al., \textit{Cooperative distributed model predictive
  control}, Systems \& Control Letters \textbf{59} (2010), no.~8, 460--469.

\bibitem{W-2009}
M.~Wooldridge, \textit{An Introduction to MultiAgent Systems}, John Wiley \&
  Sons, 2009.

\bibitem{ZLJL-2022}
Y.~Zhou et~al., \textit{Robust prescribed-time consensus of multi-agent systems
  with actuator saturation and actuator faults}, Asian Journal of Control
  \textbf{24} (2022), no.~2, 743--754.

\bibitem{GLZ-2020}
Z.~Guo, S.~Li, and Y.~Zheng, \textit{Feedback linearization based distributed
  model predictive control for secondary control of islanded microgrid}, Asian
  Journal of Control \textbf{22} (2020), no.~1, 460--473.

\bibitem{ZGLXZL-2020}
Z.~Zhao et~al., \textit{Distributed model predictive control strategy for
  islands multimicrogrids based on noncooperative game}, IEEE Transactions on
  Industrial Informatics \textbf{17} (2020), no.~6, 3803--3814.

\bibitem{HL-2021}
Y.~He and S.~Li, \textit{Distributed model predictive control with guaranteed
  performance for reconfigurable power flow systems based on passivity}, Asian
  Journal of Control \textbf{23} (2021), no.~4, 1817--1830.

\bibitem{ZSLD-2022}
B.~Zhang et~al., \textit{Distributed fault tolerant model predictive control
  for multi-unmanned aerial vehicle system}, Asian Journal of Control
  \textbf{24} (2022), no.~3, 1273--1292.

\bibitem{DZQSX-2022}
L.~Dai et~al., \textit{Distributed economic {MPC} for dynamically coupled
  linear systems: A {Lyapunov-based} approach}, IEEE Transactions on Systems,
  Man, and Cybernetics: Systems \textbf{53} (2022), no.~3, 1408--1419.

\bibitem{DQSZX-2022}
L.~Dai et~al., \textit{Distributed economic {MPC} for dynamically coupled
  linear systems with uncertainties}, IEEE Transactions on Cybernetics
  \textbf{52} (2020), no.~6, 5301--5310.

\bibitem{CMFMC-2021}
P.~Chanfreut et~al., \textit{Distributed model predictive control for tracking:
  A coalitional clustering approach}, IEEE Transactions on Automatic Control
  \textbf{67} (2021), no.~12, 6873--6880.

\bibitem{WM-2022}
Y.~Wang and C.~Manzie, \textit{Robust distributed model predictive control of
  linear systems: Analysis and synthesis}, Automatica \textbf{137} (2022),
  110141.

\bibitem{HL-2023}
W.~He and S.~Li, \textit{Enhancing topological information of the
  {Lyapunov-based} distributed model predictive control design for large-scale
  nonlinear systems}, Asian Journal of Control \textbf{25} (2023), no.~2,
  1476--1487.

\bibitem{DGPW-2016}
B.~Ding et~al., \textit{Distributed {MPC} for tracking and formation of
  homogeneous multi-agent system with time-varying communication topology},
  Asian Journal of Control \textbf{18} (2016), no.~3, 1030--1041.

\bibitem{T-2007}
P.~Tabuada, \textit{Event-triggered real-time scheduling of stabilizing control
  tasks}, IEEE Transactions on Automatic control \textbf{52} (2007), no.~9,
  1680--1685.

\bibitem{DFJ-2011}
D.~V. Dimarogonas, E.~Frazzoli, and K.~H. Johansson, \textit{Distributed
  event-triggered control for multi-agent systems}, IEEE Transactions on
  Automatic control \textbf{57} (2011), no.~5, 1291--1297.

\bibitem{ZSLNL-2019}
Y.~Zou et~al., \textit{Event-triggered distributed predictive control for
  asynchronous coordination of multi-agent systems}, Automatica \textbf{99}
  (2019), 92--98.

\bibitem{ZLXG-2022}
Y.~Zhou et~al., \textit{Event-triggered distributed robust model predictive
  control for a class of nonlinear interconnected systems}, Automatica
  \textbf{136} (2022), 110039.

\bibitem{ZJWL-2018}
J.~Zhan et~al., \textit{Distributed model predictive consensus with
  self-triggered mechanism in general linear multiagent systems}, IEEE
  Transactions on Industrial Informatics \textbf{15} (2018), no.~7, 3987--3997.

\bibitem{MZLK-2019}
X.~Mi et~al., \textit{Self-triggered {DMPC} design for cooperative multiagent
  systems}, IEEE Transactions on Industrial Electronics \textbf{67} (2019),
  no.~1, 512--520.

\bibitem{WZS-2021}
H.~Wei, K.~Zhang, and Y.~Shi, \textit{Self-triggered min--max {DMPC} for
  asynchronous multiagent systems with communication delays}, IEEE Transactions
  on Industrial Informatics \textbf{18} (2021), no.~10, 6809--6817.

\bibitem{BJ-2014}
B.~Gharesifard and J.~Cort¨¦s, \textit{Distributed continuous-time convex
  optimization on weight-balanced digraphs}, IEEE Transactions on Automatic
  Control \textbf{59} (2014), no.~3, 781--786.

\bibitem{XDLX-2021}
H.~Xie et~al., \textit{Robust {MPC} for disturbed nonlinear discrete-time
  systems via a composite self-triggered scheme}, Automatica \textbf{127}
  (2021), 109499.

\bibitem{H-2009}
J.~M. Holte, \textit{Discrete {Gronwall} lemma and applications},
  \textit{MAA-NCS meeting at the University of North Dakota}, vol.~24, 1--7.

\bibitem{KSD-2015}
K.~Hashimoto, S.~Adachi, and D.~V. Dimarogonas, \textit{Distributed aperiodic
  model predictive control for multi-agent systems}, IET Control Theory \&
  Applications \textbf{9} (2015), no.~1, 10--20.

\bibitem{MM-1993}
H.~Michalska and D.~Q. Mayne, \textit{Robust receding horizon control of
  constrained nonlinear systems}, IEEE Transactions on Automatic control
  \textbf{38} (1993), no.~11, 1623--1633.

\bibitem{GH-2006}
D.~Gu and H.~Hu, \textit{Receding horizon tracking control of wheeled mobile
  robots}, IEEE Transactions on Control Systems Technology \textbf{14} (2006),
  no.~4, 743--749.

\bibitem{KG-2002}
H.~K. Khalil and J.~W. Grizzle, \textit{Nonlinear Systems}, Prentice hall,
  2002.

\end{thebibliography}

\end{document}